\documentclass[12pt, onecolumn,a4paper]{quantumarticle}
\pdfoutput=1 
\usepackage[T1]{fontenc}

\usepackage{tikz}
\usetikzlibrary{positioning, calc}
\usetikzlibrary{calc}
\usetikzlibrary{arrows}
\usepackage{tikz-3dplot}
\usepackage{graphicx}
\usepackage{mdwlist}
\usepackage{color}
\usepackage{thmtools}
\usepackage{amssymb}
\usepackage{physics}
\usepackage{amsmath}
\usepackage{array}
\usepackage{doi}
\usepackage{url,hyperref}

\usepackage[numbers,sort&compress]{natbib}
\usetikzlibrary{fadings}
\usetikzlibrary{decorations.pathmorphing}
\usepackage{mathrsfs}
\usepackage{authblk}
\usepackage[mathscr]{euscript}
\usepackage{enumitem}
\usetikzlibrary{quantikz}

\DeclareMathAlphabet\mathbfcal{OMS}{cmsy}{b}{n}

\usepackage{cleveref}

\usetikzlibrary{decorations.pathreplacing}
\tikzset{snake it/.style={decorate, decoration=snake}}

\usetikzlibrary{decorations.pathreplacing,decorations.markings}

\tikzset{
    >=stealth',
    punkt/.style={
           rectangle,
           rounded corners,
           draw=black, very thick,
           text width=6.5em,
           minimum height=2em,
           text centered},
    pil/.style={
           ->,
           thick,
           shorten <=2pt,
           shorten >=2pt,},
  on each segment/.style={
    decorate,
    decoration={
      show path construction,
      moveto code={},
      lineto code={
        \path [#1]
        (\tikzinputsegmentfirst) -- (\tikzinputsegmentlast);
      },
      curveto code={
        \path [#1] (\tikzinputsegmentfirst)
        .. controls
        (\tikzinputsegmentsupporta) and (\tikzinputsegmentsupportb)
        ..
        (\tikzinputsegmentlast);
      },
      closepath code={
        \path [#1]
        (\tikzinputsegmentfirst) -- (\tikzinputsegmentlast);
      },
    },
  },
  mid arrow/.style={postaction={decorate,decoration={
        markings,
        mark=at position .5 with {\arrow[#1]{stealth'}}
      }}}
}

\usepackage{hyperref}
\usepackage{slashed}
\usepackage{float} 
\usepackage{graphicx} 
\usepackage{subcaption}

\usepackage{bbold}

\newtheorem{theorem}{Theorem}

\newtheorem{conjecture}[theorem]{Conjecture}
\newtheorem{corollary}[theorem]{Corollary}

\newtheorem{definition}[theorem]{Definition}

\newtheorem{lemma}[theorem]{Lemma}

\newtheorem{remark}[theorem]{Remark}

\newenvironment{proof}[1][Proof]{\noindent\textbf{#1.}}{\ \rule{0.5em}{0.5em}}
\begin{document}

\title{Relating non-local quantum computation to \newline information theoretic cryptography}

\author[3]{Rene Allerstorfer}
\email{Rene.Allerstorfer@cwi.nl}
\orcid{}

\author[3,4]{Harry Buhrman}
\email{Harry.Buhrman@cwi.nl}
\orcid{}

\author[2]{Alex May}
\email{amay@perimeterinstitute.ca}
\orcid{0000-0002-4030-5410}

\author[3,4]{Florian Speelman}
\email{f.speelman@uva.nl}
\orcid{}

\author[3]{Philip Verduyn Lunel}
\email{philip.verduyn.lunel@cwi.nl}
\orcid{}

\affiliation[2]{Perimeter Institute for Theoretical Physics}
\affiliation[3]{QuSoft, CWI Amsterdam}
\affiliation[4]{University of Amsterdam}

\abstract{
Non-local quantum computation (NLQC) is a cheating strategy for position-verification schemes, and has appeared in the context of the AdS/CFT correspondence. 
Here, we connect NLQC to the wider context of information theoretic cryptography by relating it to a number of other cryptographic primitives. 
We show one special case of NLQC, known as $f$-routing, is equivalent to the quantum analogue of the conditional disclosure of secrets (CDS) primitive, where by equivalent we mean that a protocol for one task gives a protocol for the other with only small overhead in resource costs. 
We further consider another special case of position-verification, which we call coherent function evaluation (CFE), and show CFE protocols induce similarly efficient protocols for the private simultaneous message passing (PSM) scenario. 
By relating position-verification to these cryptographic primitives, a number of results in the information theoretic cryptography literature give new implications for NLQC, and vice versa. 
These include the first sub-exponential upper bounds on the worst case cost of $f$-routing of $2^{O(\sqrt{n\log n})}$ entanglement, the first example of an efficient $f$-routing strategy for a problem believed to be outside $P/poly$, linear lower bounds on quantum resources for CDS in the quantum setting, linear lower bounds on communication cost of CFE, and efficient protocols for CDS in the quantum setting for functions that can be computed with quantum circuits of low $T$ depth.  
}
\maketitle

\pagebreak

\tableofcontents

\flushbottom

\section{Introduction}\label{sec:intro}

In a position-verification scenario, a verifier attempts to determine the location of a prover by communicating with them remotely \cite{chandran2009position,kent2011quantum,buhrman2014position}.
Position-verification may be of interest as a goal in itself, or may serve as an authentication mechanism for use towards further cryptographic goals. 
In the most widely studied setting, where the prover holds no secret key, an adversary may use a strategy known as non-local quantum computation to simulate the actions of the prover.
A non-local quantum computation replaces local actions within a designated spacetime region with actions outside that region along with entanglement shared across it.
The basic setting is shown in figure \ref{fig:non-localandlocal}. 

\begin{figure*}
    \centering
    \begin{subfigure}{0.45\textwidth}
    \centering
    \begin{tikzpicture}[scale=0.6]
    
    \draw[thick] (-1,-1) -- (-1,1) -- (1,1) -- (1,-1) -- (-1,-1);
    
    \draw[thick] (-3.5,-3) to [out=90,in=-90] (-0.5,-1);
    \draw[thick] (3.5,-3) to [out=90,in=-90] (0.5,-1);
    
    \draw[thick] (0.5,1) to [out=90,in=-90] (3.5,3);
    \draw[thick] (-0.5,1) to [out=90,in=-90] (-3.5,3);
    
    \node at (0,0) {$\mathbf{U}$};
    
    \node at (0,-5) {$ $};
    
    \end{tikzpicture}
    \caption{}
    \label{fig:local}
    \end{subfigure}
    \hfill
    \begin{subfigure}{0.45\textwidth}
    \centering
    \begin{tikzpicture}[scale=0.4]
    
    \draw[thick] (-5,-5) -- (-5,-3) -- (-3,-3) -- (-3,-5) -- (-5,-5);
    \node at (-4,-4) {$\mathbfcal{V}^L$};
    
    \draw[thick] (5,-5) -- (5,-3) -- (3,-3) -- (3,-5) -- (5,-5);
    \node at (4,-4) {$\mathbfcal{V}^R$};
    
    \draw[thick] (5,5) -- (5,3) -- (3,3) -- (3,5) -- (5,5);
    \node at (4,4) {$\mathbfcal{W}^R$};
    
    \draw[thick] (-5,5) -- (-5,3) -- (-3,3) -- (-3,5) -- (-5,5);
    \node at (-4,4) {$\mathbfcal{W}^L$};
    
    \draw[thick] (-4.5,-3) -- (-4.5,3);
    
    \draw[thick] (4.5,-3) -- (4.5,3);
    
    \draw[thick] (-3.5,-3) to [out=90,in=-90] (3.5,3);
    
    \draw[thick] (3.5,-3) to [out=90,in=-90] (-3.5,3);
    
    \draw[thick] (-3.5,-5) to [out=-90,in=-90] (3.5,-5);
    \draw[black] plot [mark=*, mark size=3] coordinates{(0,-7.05)};
    
    \draw[thick] (-4.5,-6) -- (-4.5,-5);
    \draw[thick] (4.5,-6) -- (4.5,-5);
    
    \draw[thick] (4.5,5) -- (4.5,6);
    \draw[thick] (-4.5,5) -- (-4.5,6);
    
    
    \end{tikzpicture}
    \caption{}
    \label{fig:non-localcomputation}
    \end{subfigure}
    \caption{(a) Circuit diagram showing the local implementation of a unitary in terms of a unitary $\mathbf{U}$. In position-verification, an honest prover implements the required unitary in this form. (b) Circuit diagram showing the non-local implementation of a unitary $\mathbf{U}$. $\mathbfcal{V}^L$, $\mathbfcal{V}^R$, $\mathbfcal{W}^L$, and $\mathbfcal{W}^R$ are quantum channels. The lower, bent wire represents an entangled state. In position-verification, a dishonest prover must use a circuit of this form to implement a required unitary.}
    \label{fig:non-localandlocal}
\end{figure*}
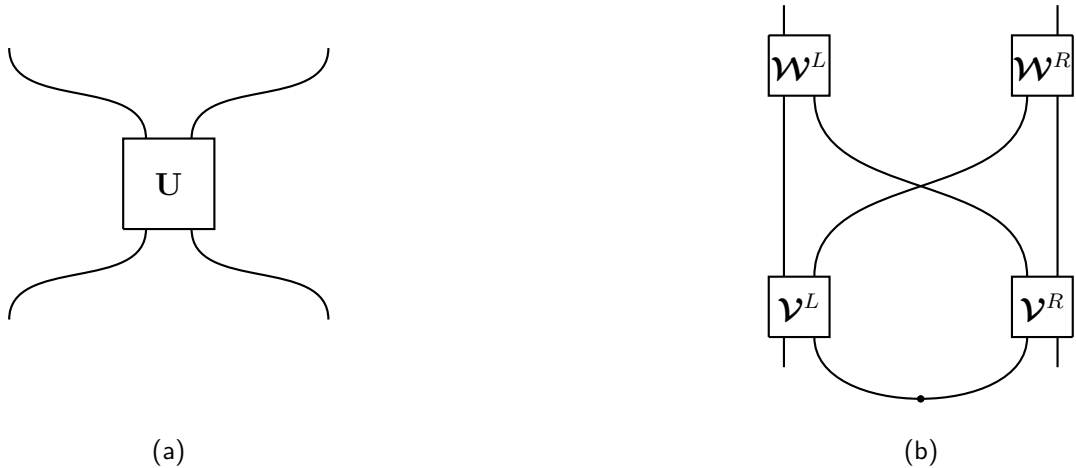

Non-local quantum computation has also been understood to arise naturally in the context of quantum gravity \cite{may2019quantum,dolev2022holography,may2023non}, in particular within the context of the AdS/CFT correspondence. 
There, a higher dimensional theory with gravity is given an equivalent description without gravity. 
In these two descriptions, processes that occur as local interactions in the higher dimensional theory are reproduced in the dual, lower dimensional description as non-local computations. 
This connection has lead to consequences for the gravitational theory \cite{may2020holographic,may2022connected}, and discussion around consequences for non-local computation  \cite{may2022complexity}. 

Because of the connections to position-verification and quantum gravity, position-verification and the related task of non-local computation have been studied by a number of authors, but basic questions remain open. 
In particular we have linear lower bounds on entanglement \cite{tomamichel2013monogamy} in a non-local computation, and exponential upper bounds \cite{beigi2011simplified}, with only a little known in between. 
For a special case of a non-local computation known as $f$-routing, where each instance is defined by a classical Boolean function $f$, the entanglement cost has been upper bounded by the size of span program computing $f$ \cite{cree2022code}, so that the class $Mod_kL$/poly\footnote{This class is reviewed in section \ref{sec:complexitymeasures}. It is inside of NC$^2$, the class of functions computed by $(\log(n))^2$ depth circuits.}  can be achieved efficiently.\footnote{This builds on earlier work \cite{buhrman2013garden} achieving the class $L$.} 
For general unitaries, Clifford unitaries can be implemented with linear entanglement, and circuits with $T$ depth of $\log n$ can be implemented with polynomial entanglement \cite{speelman2015instantaneous}. 

\begin{figure*}
    \centering
    \begin{subfigure}{0.45\textwidth}
    \centering
    \begin{tikzpicture}[scale=0.4]
    
    \draw[thick] (-5,-5) -- (-5,-3) -- (-3,-3) -- (-3,-5) -- (-5,-5);
    
    \draw[thick] (5,-5) -- (5,-3) -- (3,-3) -- (3,-5) -- (5,-5);
    
    \draw[thick] (5,5) -- (5,3) -- (3,3) -- (3,5) -- (5,5);
    
    \draw[thick] (4,-3) -- (4.5,3);
    
    \draw[thick] (-4,-3) to [out=90,in=-90] (3.5,3);
    
    \draw[thick,dashed] (-3.5,-5) to [out=-90,in=-90] (3.5,-5);
    \draw[black] plot [mark=*, mark size=3] coordinates{(0,-7.05)};
    
    \draw[thick] (-4.5,-6) -- (-4.5,-5);
    \node[below] at (-4.5,-6) {$x,s,r$};
    
    \draw[thick] (4.5,-6) -- (4.5,-5);
    \node[below] at (4.5,-6) {$y$};

    \node[below] at (3.5,-6) {$r$};

    \node[left] at (0,1) {$m_0(x,s,r)$};
    \node[right] at (4.5,0) {$m_1(y,r)$};
    
    \draw[thick] (4,5) -- (4,6);
    \node[above] at (4,6) {$s$};
    
    \end{tikzpicture}
    \caption{}
    \label{fig:CDS}
    \end{subfigure}
    \hfill
    \begin{subfigure}{0.45\textwidth}
    \centering
        \begin{tikzpicture}[scale=0.4]
    
    \draw[thick] (-5,-5) -- (-5,-3) -- (-3,-3) -- (-3,-5) -- (-5,-5);
    
    \draw[thick] (5,-5) -- (5,-3) -- (3,-3) -- (3,-5) -- (5,-5);
    
    \draw[thick] (5,5) -- (5,3) -- (3,3) -- (3,5) -- (5,5);
    
    \draw[thick] (4,-3) -- (4.5,3);
    
    \draw[thick] (-4,-3) to [out=90,in=-90] (3.5,3);
    
    \draw[thick,dashed] (-3.5,-5) to [out=-90,in=-90] (3.5,-5);
    \node[below] at (-3.5,-6) {$r$};
    \draw[black] plot [mark=*, mark size=3] coordinates{(0,-7.05)};
    \node[below] at (3.5,-6) {$r$};

    \node[left] at (0,1) {$m_0(x,r)$};
    \node[right] at (4.5,0) {$m_1(y,r)$};
    
    \draw[thick] (-4.5,-6) -- (-4.5,-5);
    \node[below] at (-4.5,-6) {$x$};
    
    \draw[thick] (4.5,-6) -- (4.5,-5);
    \node[below] at (4.5,-6) {$y$};
    
    \draw[thick] (4,5) -- (4,6);
    \node[above] at (4,6) {$f(x,y)$};
    
    \end{tikzpicture}
    \caption{}
    \label{fig:PSM}
    \end{subfigure}
    \caption{(a) A conditional disclosure of secrets (CDS) protocol. In the classical setting, Alice and Bob share randomness but do no communicate. They receive inputs $x$ and $y$ respectively. Alice additionally holds a secret $s$. They send messages to the referee. The protocol is correct if the referee can recover $s$ from the messages if and only if $f(x,y)=1$. 
    In the quantum setting, the randomness may be replaced by entanglement and the messages and secret can be quantum. 
    (b) A private simultaneous message protocol (PSM). Again Alice and Bob do not communicate but share randomness. They hold inputs $x$ and $y$ respectively. The referee should be able to learn $f(x,y)$ but nothing else about $(x,y)$. In the quantum setting the randomness is replaced with entanglement, and the messages can be quantum.}
    \label{fig:PSMandCDS}
\end{figure*}
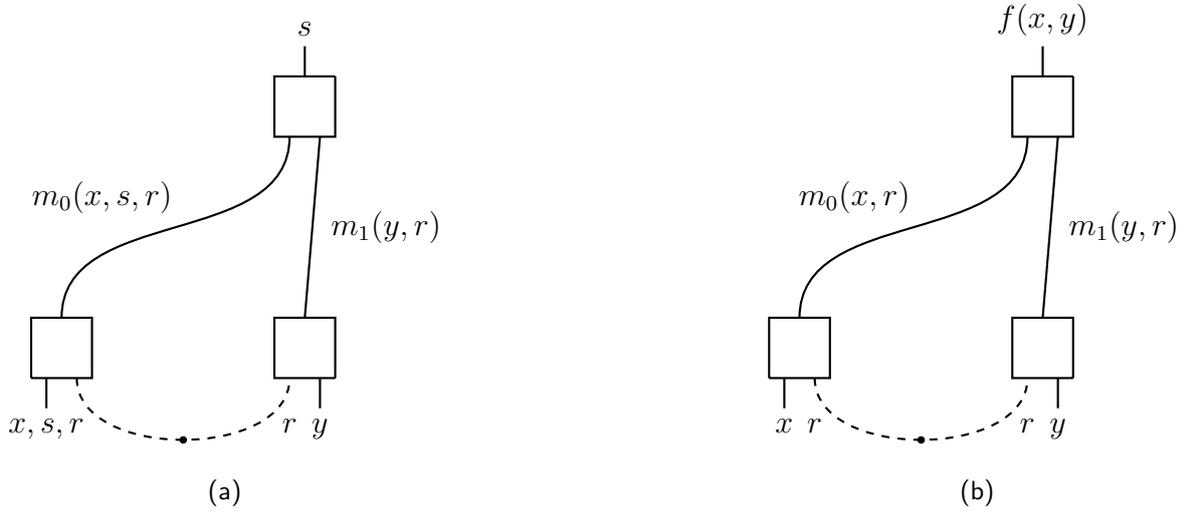

In this article we prove connections between two well studied cryptographic primitives, conditional disclosure of secrets (CDS) \cite{GERTNER2000592} and private simultaneous message passing (PSM) \cite{ishai1997private}, and non-local quantum computation. 
These primitives are studied in the context of information theoretic cryptography, in particular in their relationship to secure multiparty computation \cite{aiello2001priced,ishai2010secure}, private information retrieval \cite{GERTNER2000592}, secret sharing \cite{applebaum2020power}, and other cryptographic goals \cite{beimel2018complexity}. 
We illustrate their functionality in figure \ref{fig:PSMandCDS}. 
Both settings generally involve $k$ parties along with a referee, but in this work we focus on the $k=2$ case, which is the setting we relate to non-local computation. 
In CDS, two non-communicating parties, Alice and Bob, receive inputs $x$ and $y$ respectively. Alice additionally holds a secret $s$. 
Alice and Bob compute messages $m_0(x,s,r)$ and $m_1(y,r)$ based on their inputs and shared randomness, which are then sent to the referee. 
The referee should be able to recover the secret $s$ if and only if $f(x,y)=1$. 
PSM is a similar setting. 
There, Alice and Bob have inputs $x$ and $y$ along with shared randomness. 
They send messages $m_{0}(x,r)$ and $m_1(y,r)$ to the referee. 
The referee should be able to compute $f(x,y)$ from the messages, but not learn anything else about the inputs $(x,y)$ than is implied by the value of $f(x,y)$. 
We give formal definitions of both primitives in section \ref{sec:primitivedefinitions}. 

To relate these primitives to non-local computation, we first show that the natural quantum generalization of CDS, which we denote as conditional disclosure of quantum secrets (CDQS), is equivalent to the $f$-routing task.
More specifically, protocols for CDQS induce similarly efficient protocols for $f$-routing, and vice versa. 
Further, we show that classical CDS protocols induce similarly efficient quantum protocols. 
We also introduce a special case of non-local quantum computation known as a coherent function evaluation (CFE), which we show is closely related to the PSM model: efficient CFE protocols induce efficient PSM protocols using quantum resources (PSQM).
We also give a weak converse that shows good PSQM protocols induce CFE protocols that succeed with constant (independent of the input size) probability.\footnote{We only prove this in the exact setting, while all other implications allow for small errors in correctness and small security leakage.}
The status of the relationship among these primitives is shown in figure \ref{fig:web}.

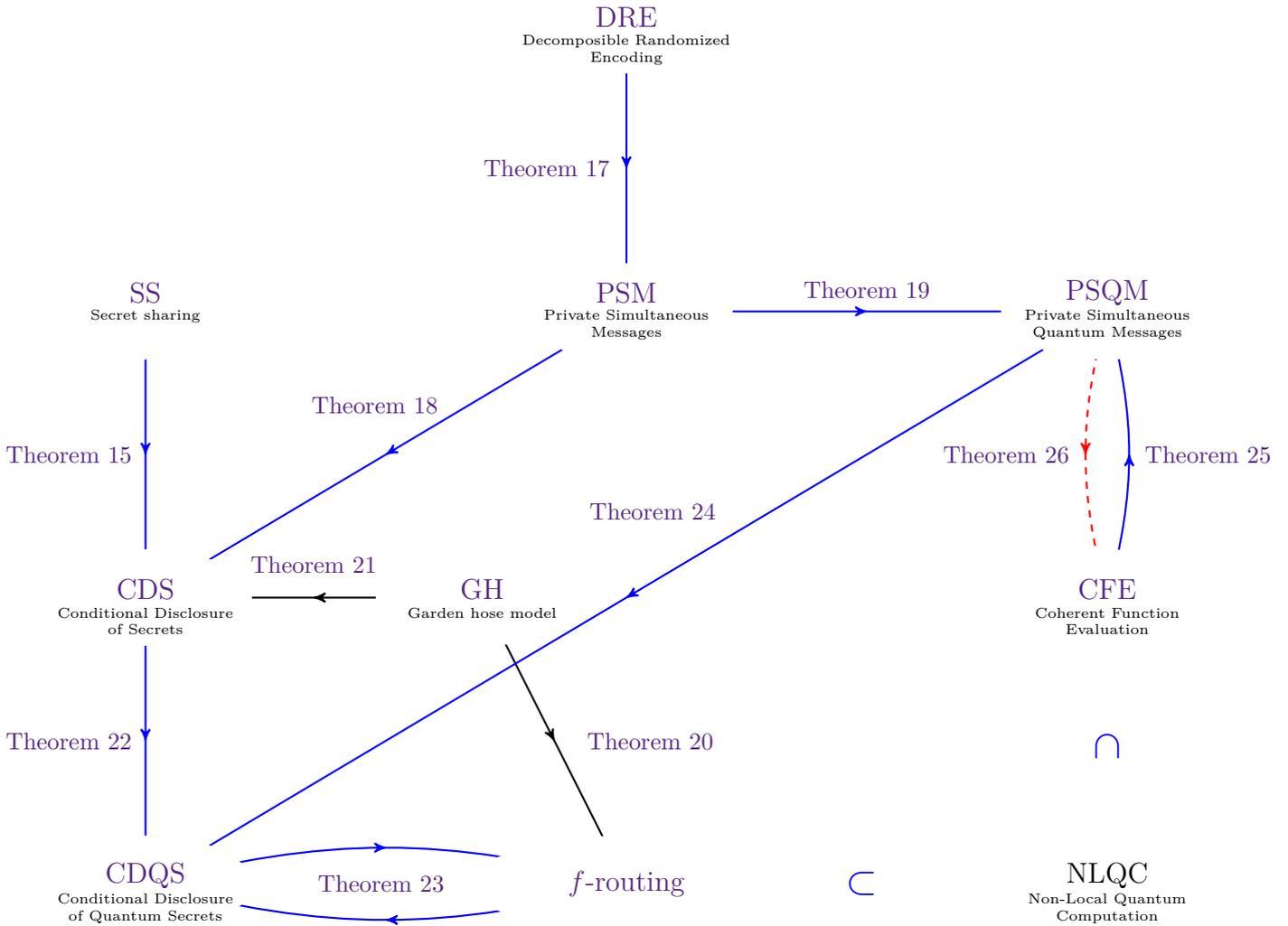
\begin{figure}
    \centering
    \begin{tikzpicture}[scale=1.4]

    \coordinate (CDS) at (0,3);
    \coordinate (PSQM) at (10,6);
    \coordinate (PSM) at (5,6);
    \coordinate (CDQS) at (0,0);
    \coordinate (fR) at (5,0); 
    \coordinate (NLQC) at (10,0);
    \coordinate (DRE) at (5,9);
    \coordinate (CFE) at (10,3);
    \coordinate (SS) at (0,6);
    \coordinate (QSS) at (-3,3);
    \coordinate (GH) at (3.5,3);

    \draw[thick, mid arrow] (GH) -- (fR);
    \node at (5.25,1.5) {\hyperref[thm:GHtofR]{\footnotesize Theorem~\ref*{thm:GHtofR}}};

    \draw[thick, mid arrow] (GH) to[out=180,in=0] (CDS);
    \node[font=\footnotesize] at (1.75,3.35) {\hyperref[thm:GHtoCDS]{\footnotesize Theorem~\ref*{thm:GHtoCDS}}};

    \draw[thick, mid arrow,blue] (CDQS) to[out=15,in=165] (fR);
    \node at (2.45,0) {\hyperref[thm:CDQSandfRouting]{\footnotesize Theorem~\ref*{thm:CDQSandfRouting}}};
    \draw[thick, mid arrow,blue] (fR) to[out=-165,in=-15] (CDQS);
    
    \draw[thick, mid arrow,blue] (PSM) -- (CDS) node[midway, left, font=\footnotesize, xshift=5ex, yshift=4ex] {\hyperref[thm:PSMgivesCDS]{Theorem~\ref*{thm:PSMgivesCDS}}};
    
    \draw[thick, mid arrow,blue] (PSM) -- (PSQM) node[midway, above, font=\footnotesize, yshift=0.25ex] {\hyperref[thm:PSMgivesPSQM]{Theorem~\ref*{thm:PSMgivesPSQM}}};;

    \path[draw=none, blue] (fR) edge[draw=none, transform shape] node[right, xshift=-2ex] {\rotatebox{0}{$\subset$}} (NLQC);

    \draw[thick, mid arrow, blue] (CDS) -- (CDQS) node[midway, left
    , font=\footnotesize, xshift=-0.25ex] {\hyperref[thm:CDStoCDQS]{Theorem~\ref*{thm:CDStoCDQS}}};

    \draw[thick, mid arrow,blue] (PSQM) -- (CDQS) node[midway, left, font=\footnotesize, xshift=8ex, yshift=7ex] {\hyperref[thm:PSQMtoCDQS]{Theorem~\ref*{thm:PSQMtoCDQS}}};

    \draw[thick, mid arrow,blue] (DRE) -- (PSM) node[midway, left
    , font=\footnotesize, xshift=-0.5ex] {\hyperref[thm:DREgivesPSM]{Theorem~\ref*{thm:DREgivesPSM}}};

    \path[draw=none, blue] (NLQC) edge[draw=none, transform shape] node[above, yshift=-2ex] {\rotatebox{270}{$\subset$}} (CFE);
    
    \draw[thick, mid arrow,blue] (CFE) to[out=75,in=-75] (PSQM) node at (11.05,4.5) {\hyperref[thm:CFEtoPSQM]{\footnotesize Theorem~\ref*{thm:CFEtoPSQM}}};
    
    \draw[thick, mid arrow,dashed,red] (PSQM) to[out=-105,in=105] (CFE) node at (8.95,4.5) {\hyperref[thm:PSQMtoCFEWeak]{\footnotesize Theorem~\ref*{thm:PSQMtoCFEWeak}}};

    \draw[thick, mid arrow,blue] (SS) -- (CDS) node[midway, left, xshift=-0.25ex] {\hyperref[thm:SSgivesCDS]{\footnotesize Theorem~\ref*{thm:SSgivesCDS}}};

    \filldraw[color=white, fill=white] (fR) ellipse (1.6 and 0.5);
    \filldraw[color=white, fill=white] (CDS) ellipse (1.1 and 0.5);
    \filldraw[color=white, fill=white] (CDQS) ellipse (1.1 and 0.5);
    \filldraw[color=white, fill=white] (CDS) ellipse (1.1 and 0.5);
    \filldraw[color=white, fill=white] (PSM) ellipse (1.1 and 0.5);
    \filldraw[color=white, fill=white] (PSQM) ellipse (1.1 and 0.5);
    \filldraw[color=white, fill=white] (NLQC) ellipse (1.1 and 0.5);
    \filldraw[color=white, fill=white] (DRE) ellipse (1.1 and 0.5);
    \filldraw[color=white, fill=white] (CFE) ellipse (1.1 and 0.5);
    \filldraw[color=white, fill=white] (SS) ellipse (1.1 and 0.5);
    \filldraw[color=white, fill=white] (GH) ellipse (1.1 and 0.5);

    \node at ([yshift=1.1ex] SS) {\hyperref[def:SS]{SS}};
    \node[below, align=center, font=\tiny, yshift=0.9ex] at (SS) {Secret sharing};

    \node at ([yshift=1.1ex] PSQM) {\hyperref[def:PSQM]{PSQM}};
    \node[below, align=center, font=\tiny, yshift=0.9ex] at (PSQM) {Private Simultaneous \\ Quantum Messages};
    
    \node at ([yshift=0.5ex] CDS) {\hyperref[def:CDS]{CDS}};
    \node[below, align=center, font=\tiny, yshift=0ex] at (CDS) {Conditional Disclosure \\ of Secrets};
    
    \node at ([yshift=1.1ex] PSM) {\hyperref[def:PSM]{PSM}};
    \node[below, align=center, font=\tiny, yshift=0.9ex] at (PSM) {Private Simultaneous \\ Messages};
        
    \node at ([yshift=0.5ex] CDQS) {\hyperref[def:CDQS]{CDQS}};
    \node[below, align=center, font=\tiny, yshift=0ex] at (CDQS) {Conditional Disclosure \\ of Quantum Secrets};
    
    \node at ([yshift=0.5ex] NLQC) {NLQC};
    \node[below, align=center, font=\tiny, yshift=0ex] at (NLQC) {Non-Local Quantum \\ Computation};
    
    \node at ([yshift=0.5ex] DRE) {\hyperref[def:DRE]{DRE}};
    \node[below, align=center, font=\tiny, yshift=0ex] at (DRE) {Decomposible Randomized \\ Encoding};
    
    \node at (fR) {\hyperref[def:frouting]{$f$-routing}};
    
    \node at ([yshift=0.5ex] CFE) {\hyperref[def:CFE]{CFE}};
    \node[below, align=center, font=\tiny, yshift=0ex] at (CFE) {Coherent Function \\ Evaluation};


    \node at ([yshift=0.5ex] GH) {\hyperref[def:gh]{GH}};
    \node[below, align=center, font=\tiny, yshift=0ex] at (GH) {Garden hose model};
            
    \end{tikzpicture}
    \caption{Implications among primitives: an arrow from X to Y says that a protocol for X implies a protocol for Y with the same efficiencies (up to constant overheads). All implications shown in blue hold in the robust setting where we allow small errors and leakages. The dashed red line indicates that a perfect PSQM protocol that succeeds with high probability implies a CFE protocol that succeeds with constant probability. The subset symbol $\subset$ indicates that $f$-routing and CFE are special cases of NLQC. Primitive abbreviations (DRE, PSM, ...) and theorem numbers link to relevant proofs or definitions.}
    \label{fig:web}
\end{figure}

Our results relate position-verification to the wider setting of information-theoretic cryptography.
This provides a partial explanation of the difficulty of finding better upper and lower bounds in non-local computation, since we now see that doing so would resolve other long-standing questions in cryptography\footnote{For example lower bounds on $f$-routing give lower bounds on (classical) CDS.}. 
In a positive direction, we use results in NLQC to give new results on CDS and PSM, and vice versa. 
Our key results are,
\begin{itemize}
    \item Sub-exponential upper bounds on entanglement cost in $f$-routing for an arbitrary function (corollary \ref{corollary:subexpfroute})
    \item Efficient CDQS and $f$-routing protocols for the quadratic residuosity problem, the first problem not known to be in P/poly with an efficient non-local computation protocol (corollaries \ref{corolary:CDSandCDQSoutsideP} and \ref{corollary:fRouteoutsideP})
\end{itemize}
These results represent significant changes in our understanding of the efficiency of $f$-routing protocols. 
Previously the best upper bounds for arbitrary functions were exponential, and the highest complexity functions with known efficient schemes were in Mod$_k$L/poly. 

From our connections between CDS, PSM, and NLQC, we also obtain a number of other implications, 
\begin{itemize}
    \item Linear lower bounds on communication complexity in CFE (corollary \ref{corollary:CFElowerbound})
    \item Linear lower bounds on entanglement in CDQS and PSQM for random functions (corollaries \ref{corollary:CDQSlowerbound} and \ref{corollary:PSQMlowerbound}), and logarithmic lower bounds on entanglement for the inner product function (corollary \ref{corollary:CDQSlogbound} and \ref{corollary:PSQMlogbound})
    \item An entanglement efficient protocol for CDQS and PSQM when the target function $f$ can be evaluated by a quantum circuit with low $T$-depth (corollaries \ref{corollary:Tdepthandf-route} and \ref{corollary:PSQMandTdepth})
\end{itemize}
More broadly our results take position-verification from being an `island' in the space of cryptographic primitives, with no known classical analogues or connections to other more standard notions, to being richly connected to a web of interrelated primitives, which themselves are related to central goals in information theoretic cryptography. 
We hope these results lead to new perspectives on position-verification, and new perspectives in the study of CDS, PSM and related primitives. 
In particular a number of classical results on CDS and PSM may find natural quantum extensions in the context of NLQC.
In the discussion we comment on some cases where quantum analogues in the NLQC setting of classical cryptographic results are not yet known. 

\vspace{0.2cm}
\noindent \textbf{Outline of this article} 
\vspace{0.2cm}

In section \ref{sec:background}, we present some relevant background. 
Section \ref{sec:QItools} gives a summary of the quantum information tools we exploit. 
Section \ref{sec:primitivedefinitions} summarizes the various cryptographic primitives which we study and relate. 
Section \ref{sec:existingrelationships} gives the already known relations among these primitives. 

In section \ref{eq:newrelationships} we prove new relationships among our set of cryptographic primitives. 
The full set of connections is presented as figure \ref{fig:web}. 
The relationships between CDS and CDQS, CDQS and $f$-routing, CFE and PSQM, and CDQS and PSQM are new to the best of our knowledge. 

In section \ref{sec:complexity} we summarize the known results on the complexity of efficiently achievable functions in the PSQM, CDQS and $f$-routing settings. 
The status of the complexity of efficiently achievable functions in the general case is not too changed by our results: existing CDS protocols give $f$-routing protocols, but in the existing literature on both $f$-routing and CDS the most efficient protocols have a cost like $(\log p) \cdot SP_p(f)$ where $SP(f)$ denotes the minimal size of a span program over $\mathbb{Z}_p$ computing $f$ \cite{GERTNER2000592,cree2022code}.

Sections \ref{sec:newlowerbounds} and \ref{sec:newprotocols} spell out the implications for non-local computation and its special cases that follow from known results in CDS and PSM, and conversely the implications for CDS and PSM that follow from known results in non-local computation. 
In section \ref{sec:newlowerbounds} we give new lower bounds that follow in this way. 
In section \ref{sec:newprotocols} we give new upper bounds. 
Our new upper bounds include the most significant implications that follow from the connections we find, which are sub-exponential upper bounds on $f$-routing for arbitrary functions and an efficient scheme for a function believed to be outside of $P/poly$. 

Section \ref{sec:discussion} concludes with some discussion and open problems, in particular commenting on connections to quantum gravity and to some results in the classical cryptography literature that may have quantum analogues relevant to the NLQC setting.  

\section{Background}\label{sec:background}

\subsection{Tools from quantum information theory}\label{sec:QItools}

In this section we briefly recall some standard tools of quantum information theory. 
We follow the conventions of \cite{wilde2013quantum}, where an overview of these tools and further references can also be found. 

\vspace{0.2cm}
\noindent \textbf{Probability distributions}
\vspace{0.2cm}

Given a random variable $X$, we label a probability distribution of $X$ by $P_X$. 
For the distribution of $X$ conditioned on $Y$, we use $P_{X|Y}$. 
When the conditioning distribution $Y$ takes the value $y\in Y$ we denote the resulting distribution on $X$ by $P_{X|y=Y}\equiv P_{X|y}$. 

\vspace{0.2cm}
\noindent \textbf{Quantum one-time pad}
\vspace{0.2cm}

The quantum one-time pad \cite{ambainis2000private} uses classical randomness to conceal quantum information. 
To understand this, suppose that Alice wishes to give Bob a quantum system $B$, but wants Bob to only obtain $B$ if he also knows a classical key $k$. 
Supposing that $B$ consists of qubits, Alice can do this by applying a random Pauli string $P^k_B$. 
If Bob does not know $k$, $B$ is hidden to him since
\begin{align}
    \frac{1}{2^{|k|}} \sum_k P^k_B \rho_{AB} P^k_B = \rho_A \otimes \frac{\mathcal{I}_B}{d_B}
\end{align}
where the index $k$ ranges over all choices of Pauli strings, and $\mathcal{I}$ represents the identity operator. 
On the other hand, if Bob knows $k$ he can undo the Pauli string and recover the $B$ system. 

\vspace{0.2cm}
\noindent \textbf{Distance measures and inequalities}
\vspace{0.2cm}

Let $\mathcal{D}(\mathcal{H}_A)$ be the set of density matrices on the Hilbert space $\mathcal{H}_A$. 
Given two density matrices $\rho$, $\sigma\in \mathcal{D}(\mathcal{H}_A)$,  define the fidelity,
\begin{align}
    F(\rho,\sigma) \equiv \left( \tr\left(\sqrt{\sqrt{\rho}\,\sigma\sqrt{\rho}}\right)\right)^2
\end{align}
which is related to the one norm distance $||\rho-\sigma||_1$ by the Fuchs van de Graff inequalities, 
\begin{align}
    1- \sqrt{F(\rho,\sigma)} \leq \frac{1}{2}||\rho-\sigma||_1 \leq \sqrt{1-F(\rho,\sigma)}.
\end{align}

It will also be useful to define the diamond norm distance, which is a distance measure on the space of quantum channels. 
\begin{definition} Let $\mathbfcal{N}_{B\rightarrow C}, \mathbfcal{M}_{B\rightarrow C}: \mathcal{L}(\mathcal{H}_A)\rightarrow \mathcal{L}(\mathcal{H}_B)$ be quantum channels. 
The \textbf{diamond norm distance} is defined by 
\begin{align}
    ||\mathbfcal{N}_{B\rightarrow C}-\mathbfcal{M}_{B\rightarrow C}||_\diamond = \sup_{d} \max_{\rho_{A_dB}}||\mathbfcal{N}_{B\rightarrow C}(\Psi_{A_dB}) - \mathbfcal{M}_{B\rightarrow C}(\Psi_{A_dB})||_1
\end{align}
where $\rho_{A_dB}\in \mathcal{D}(\mathcal{H}_{A_d}\otimes \mathcal{H}_B)$ and $\mathcal{H}_{A_d}$ is a $d$ dimensional Hilbert space. 
\end{definition}
The diamond norm distance has an operational interpretation in the terms of the maximal probability of distinguishing quantum channels \cite{kitaev2002classical,wilde2013quantum}. 

\vspace{0.2cm}
\noindent \textbf{Decoupling and recovery}
\vspace{0.2cm}

The basic idea underlying the connection between CDS and $f$-routing that we will give is the notion of decoupling and complementary recovery. 
To develop this, consider a quantum channel $\mathbfcal{N}_{B\rightarrow C}: \mathcal{L}(\mathcal{H}_B)\rightarrow \mathcal{L}(\mathcal{H}_C)$. 
We would like to understand when this channel has an (approximate) inverse. 
Consider any unitary extension of the channel, call it $\mathbf{V}_{BE'\rightarrow CE}$, which satisfies
\begin{align}
    \mathbfcal{N}_{B\rightarrow C}(\cdot) = \tr_{E} (\mathbf{V}_{BE'\rightarrow CE} (\cdot) \mathbf{V}^\dagger_{BE'\rightarrow CE}).
\end{align}
A classic result \cite{schumacher2002entanglement,schumacher2002approximate} says if we input a maximally entangled state $\ket{\Psi^+}_{AB}$ and find that $I(A:E)_{\mathbfcal{N}(\Psi^+)}$ is small, say less than $\epsilon$, then there exists an inverse channel $\mathbfcal{N}^{-1}_{B\rightarrow C}$ which works well in the sense that the fidelity
\begin{align} F(\Psi^+,\mathbfcal{N}^{-1}_{B\rightarrow C}\circ \mathbfcal{N}_{B\rightarrow C}(\Psi^+)) \geq 1-\sqrt{\epsilon}.
\end{align}
The inverse channel is succeeding when acting on the maximally entangled state, which can also be understood as acting correctly in an averaged (over input states) sense. 

We will make use of a stronger notion of decoupling, which shows that a worst case notion of decoupling implies the existence of an inverse channel that always succeeds. 
The theorem was proved in \cite{kretschmann2008information}.
\begin{theorem}\label{thm:newdecoupling}
    Let $\mathbfcal{N}_{A\rightarrow B}:\mathcal{L}(\mathcal{H}_A)\rightarrow \mathcal{L}(\mathcal{H}_B)$ be a quantum channel, and let $\mathbfcal{N}^c_{A\rightarrow E}$ be the complimentary channel. 
    Let $\mathbfcal{S}_{A \rightarrow E}$ be a completely depolarizing channel, which traces out the input and replaces it with a fixed state $\sigma_E$. 
    Then we have that
    \begin{align}
        \frac{1}{4}\inf_{\mathbfcal{D}_{B\rightarrow A}} || \mathbfcal{D}_{B\rightarrow A} \circ \mathbfcal{N}_{A\rightarrow B}-\mathbfcal{I}_{A\rightarrow A}||_{\diamond}^2 \leq ||\mathbfcal{N}^c_{A\rightarrow E} - \mathbfcal{S}_{A \rightarrow E} ||_\diamond \leq 2 \inf_{\mathbfcal{D}_{B\rightarrow A}} ||\mathbfcal{D}_{B\rightarrow A} \circ \mathbfcal{N}_{A\rightarrow B} - \mathbfcal{I}_{A\rightarrow A} ||_\diamond^{1/2}\nonumber 
    \end{align}
    where the infimum is over all quantum channels $\mathbfcal{D}_{B\rightarrow A}$. 
\end{theorem}

The above should be understood as saying that if there is a good inverse to the channel $\mathbfcal{N}_{A\rightarrow B}$, then the complementary channel is close to depolarizing, and vice versa. 
Intuitively, the depolarizing channel reveals no information about $A$, so this is saying the existence of an inverse is equivalent to not leaking information to the environment. 

\subsection{Definitions of the primitives}\label{sec:primitivedefinitions}

In this section we give the definitions of each of the primitives that we discuss in this article. 
Note that we focus on information theoretic definitions of security. 
In all cases there are meaningful versions of these primitives with computational security, but we have not explored their connections to non-local computation. 

\vspace{0.2cm}
\noindent \textbf{Conditional disclosure of secrets}
\vspace{0.2cm}

We first define the classical CDS setting, which we also illustrate in figure \ref{fig:CDS}. 
\begin{definition}\label{def:CDS}
    A \textbf{conditional disclosure of secrets (CDS)} task is defined by a choice of function $f:\{0,1\}^{2n}\rightarrow \{0,1\}$.
    The scheme involves inputs $x\in \{0,1\}^{n}$ given to Alice, and input $y\in \{0,1\}^{n}$ given to Bob.
    Alice and Bob share a random string $r\in R$.
    Additionally, Alice holds a string $s$ drawn from distribution $S$, which we call the secret. 
    Alice sends message $m_0(x,s,r)\in M_0$ to the referee, and Bob sends message $m_1(y,r)\in M_1$.  
    We require the following two conditions on a CDS protocol. 
    \begin{itemize}
        \item $\epsilon$\textbf{-correct:} There exists a decoding function $D(m_0,x,m_1,y)$ such that 
        \begin{align}
            \forall s\in S,\,\forall \,(x,y) \in X\times Y \,\,s.t.\,\,f(x,y)=1,\,\,\, \underset{r\leftarrow R}{\mathrm{Pr}}[D(m_0,x,m_1,y)=s] \geq 1-\epsilon
        \end{align}
        \item $\delta$\textbf{-secure:} There exists a simulator producing a distribution $Sim$ on the random variable $M=M_0M_1$ such that
        \begin{align}
            \forall s\in S,\,\forall \,(x,y) \in X\times Y \,\,s.t.\,\, f(x,y)=0,\,\,\, ||Sim_{M|xy} - P_{M|xys} ||_1\leq \delta
        \end{align}
    \end{itemize}
\end{definition}

Notice that in our definition of CDS we have imposed that the secret be held only by Alice. 
We can easily transform protocols that succeed with the secret held on both sides to one where the secret is held only on one side. 
This is a standard remark about CDS, though we don't know a reference where this is shown in the imperfect setting, so we give the simple proof of this fact here.  
\begin{remark}\label{remark:onesidedCDS}
A CDS task where $s$ is initially held by Alice and Bob can be turned into one where only Alice holds $s$ at the cost of $|s|$ shared random bits, and $|s|$ bits of communication. 
If the CDS protocol is $\epsilon$-correct and $\delta$-secure, the one-sided protocol will be $\epsilon$-correct and $O(\delta)$ secure.
\end{remark}
\begin{proof}\,
    To see this, suppose we have a perfectly correct and secure CDS protocol which works when $s$ is held on both sides. 
    Then run this protocol on a randomly chosen $s'$, and have Alice send $s'\oplus s$ to the referee. 
    Only Alice needs to know $s$ to run this protocol. 
    
    Suppose our initial CDS protocol is $\epsilon$-correct and $\delta$-secure. 
    Then the new CDS will also be $\epsilon$-correct, since $s$ can be computed deterministically from $s'$ and the bit $\tilde{s}=s\oplus s'$. 
    To understand security, note that $\delta$-security of the original protocol implies
    \begin{align}
        ||P_{S'M}-P_{S'}P_{M} ||_1\leq \delta
    \end{align}
    Using this, $P_{S\tilde{S}}=P_{S}P_{\tilde{S}}$ (from the properties of the one-time pad), and that $S$ and $M$ are independent conditioned on $\tilde{S}$, we have
    \begin{align}
        ||P_{S\tilde{S}M} - P_{S}P_{\tilde{S}}P_{M}||_1&= ||P_{S|\tilde{S}M}P_{\tilde{S}M} - P_{S}P_{\tilde{S}}P_{M}||_1\nonumber \\
        &= ||P_{S|\tilde{S}}P_{\tilde{S}M} - P_{S}P_{\tilde{S}}P_{M}||_1\nonumber \\
        &\leq ||P_{S|\tilde{S}}P_{\tilde{S}}P_{M} - P_{S}P_{\tilde{S}}P_{M}||_1+ \delta \nonumber \\
        &= ||P_{S\tilde{S}}P_{M} - P_{S}P_{\tilde{S}}P_{M}||_1+ \delta \nonumber\\
        &= \delta
    \end{align}
    which is exactly $\delta$ security of the one sided CDS protocol. 
\end{proof}

Finally, we remark that a CDS for secret $s_1$ and a CDS for secret $s_2$ can be run in parallel using fresh randomness while maintaining security and correctness of each CDS scheme. 
To see this, call the message for the first CDS $M_1$ and the message for the second CDS $M_2$. 
If we consider how much the referee can learn about the secret $s_1$, message $M_2$ doesn't reveal anything, because it depends only on the randomness $r_2$, the inputs (which the referee knows already as part of the CDS for $s_1$), and $s_2$. 
All of these variables are already known by the referee as part of the CDS for $s_1$, or are uncorrelated with $s_1$. 
More succinctly, the distribution on $s_1$ is independent of $M_2$ when conditioning on $XY$, so revealing $M_2$ doesn't help the referee learn $s_1$, given that they already know $XY$, or in notation
\begin{align}\label{eq:independentCDS}
    P_{M_1M_2|xys} = P_{M_1|xys_1}P_{M_2|xys_2}
\end{align}
A similar statement establishes security of the CDS hiding $s_2$ in the presence of message $M_1$. 

As a consequence of the above comments, the CDS hiding secret $s=(s_1,s_2)$ given by running the CDS for each secret in parallel has good security and correctness, as we capture in the next lemma.\footnote{This is a simple, but weak, method of obtaining a CDS for a long secret from CDS for a short secret. It will suffice for our purposes, but see \cite{applebaum2017conditional,applebaum2020power} for improved results.} 
\begin{lemma}
    Suppose we have a CDS for function $f$ which is $\epsilon$-correct and $\delta$-secure, and hides $k$ bits, and uses $r$ bits of randomness and $c$ bits of communication. 
    Then we can build a CDS for function $f$ that hides $mk$ bits, is $m\epsilon $ correct and $m\delta$ secure and which uses $m r$ bits of randomness and $m c$ bits of communication.
\end{lemma}
\begin{proof}\, The strategy is to repeat the CDS protocol that hides $k$ bits $m$ times in parallel. 
To understand correctness of the new protocol, notice that on $1$ instances the probability of the referee guessing $s_i$ correctly is at least $1-\epsilon$, so their probability of guessing all $m$ strings $s_i$ correctly is at least $(1-\epsilon)^m \geq (1-mk)$. 
To understand security, we define a simulator for the composed protocol by taking the product of the distributions for a single instance of the protocol, 
\begin{align}
    Sim_{M_1...M_m|xy} \equiv Sim_{M_1|xy}...Sim_{M_m|xy}.
\end{align}
We also note that, using fresh randomness for each instance of the CDS, we can extend equation \ref{eq:independentCDS} to 
\begin{align}
    P_{M_1...M_m|xys} = P_{M_1|xys_1}...P_{M_m|xys_m}.
\end{align}
Then by repeated application of the triangle inequality, and using security of each instance of the CDS, we have that on $0$ instances
\begin{align}
    ||Sim_{M_1...M_m|xy} - P_{M_1...M_m|xys}||_1  = ||Sim_{M_1|xy}...Sim_{M_m|xy} - P_{M_1|xys_1}...P_{M_m|xys}||_1 \leq m\delta\nonumber
\end{align}
as claimed. 
\end{proof}

\vspace{0.2cm}
\noindent \textbf{Conditional disclosure of quantum secrets}
\vspace{0.2cm}

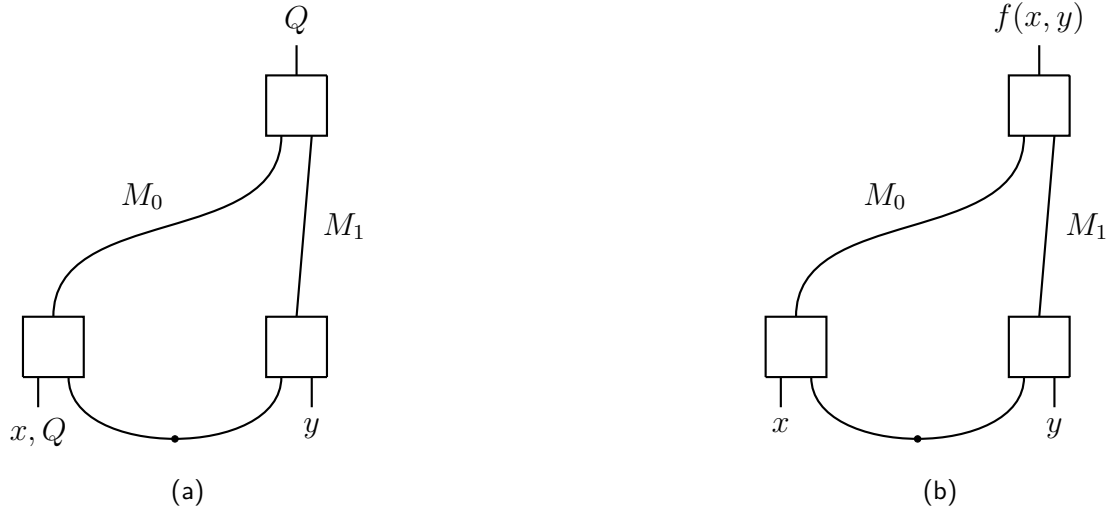
\begin{figure*}
    \centering
    \begin{subfigure}{0.45\textwidth}
    \centering
    \begin{tikzpicture}[scale=0.4]
    
    \draw[thick] (-5,-5) -- (-5,-3) -- (-3,-3) -- (-3,-5) -- (-5,-5);
    
    \draw[thick] (5,-5) -- (5,-3) -- (3,-3) -- (3,-5) -- (5,-5);
    
    \draw[thick] (5,5) -- (5,3) -- (3,3) -- (3,5) -- (5,5);
    
    \draw[thick] (4,-3) -- (4.5,3);
    
    \draw[thick] (-4,-3) to [out=90,in=-90] (3.5,3);
    
    \draw[thick] (-3.5,-5) to [out=-90,in=-90] (3.5,-5);
    \draw[black] plot [mark=*, mark size=3] coordinates{(0,-7.05)};
    
    \draw[thick] (-4.5,-6) -- (-4.5,-5);
    \node[below] at (-4.5,-6) {$x,Q$};
    
    \draw[thick] (4.5,-6) -- (4.5,-5);
    \node[below] at (4.5,-6) {$y$};

    \node[left] at (0,1) {$M_0$};
    \node[right] at (4.5,0) {$M_1$};
    
    \draw[thick] (4,5) -- (4,6);
    \node[above] at (4,6) {$Q$};
    
    \end{tikzpicture}
    \caption{}
    \label{fig:CDQS}
    \end{subfigure}
    \hfill
    \begin{subfigure}{0.45\textwidth}
    \centering
        \begin{tikzpicture}[scale=0.4]
    
    \draw[thick] (-5,-5) -- (-5,-3) -- (-3,-3) -- (-3,-5) -- (-5,-5);
    
    \draw[thick] (5,-5) -- (5,-3) -- (3,-3) -- (3,-5) -- (5,-5);
    
    \draw[thick] (5,5) -- (5,3) -- (3,3) -- (3,5) -- (5,5);
    
    \draw[thick] (4,-3) -- (4.5,3);
    
    \draw[thick] (-4,-3) to [out=90,in=-90] (3.5,3);
    
    \draw[thick] (-3.5,-5) to [out=-90,in=-90] (3.5,-5);
    \draw[black] plot [mark=*, mark size=3] coordinates{(0,-7.05)};

    \node[left] at (0,1) {$M_0$};
    \node[right] at (4.5,0) {$M_1$};
    
    \draw[thick] (-4.5,-6) -- (-4.5,-5);
    \node[below] at (-4.5,-6) {$x$};
    
    \draw[thick] (4.5,-6) -- (4.5,-5);
    \node[below] at (4.5,-6) {$y$};
    
    \draw[thick] (4,5) -- (4,6);
    \node[above] at (4,6) {$f(x,y)$};
    
    \end{tikzpicture}
    \caption{}
    \label{fig:PSQMintro}
    \end{subfigure}
    \caption{(a) Illustration of a CDQS protocol. Alice and Bob share an entangled resource state, illustrated as the solid curved line. Alice receives the classical string $x\in \{0,1\}^{n}$ as input, and a quantum system $Q$, which we take to be maximally entangled with a reference $R$. Bob receives input $y\in \{0,1\}^{n}$. Alice and Bob prepare quantum systems $M_0$ and $M_1$, which they pass to the referee. The protocol is correct if when $f(x,y)=1$ the map from $Q$ to $M_0M_1$ can be reversed, and secure when for $f(x,y)=0$ the $M=M_0M_1$ system is independent of the input state on $Q$. See definition \ref{def:CDQS}. (b) A PSQM protocol. Again Alice and Bob share an entangled resource state. Alice receives input $x\in \{0,1\}^{n}$, Bob receives input $y\in \{0,1\}^{n}$.  Alice and Bob prepare quantum systems $M_0$ and $M_1$, which they pass to the referee. The protocol succeeds if the referee can determine $f(x,y)$, but the system $M=M_0M_1$ otherwise reveals nothing about the inputs $x,y$. See definition \ref{def:PSQM}.}
    \label{fig:PSQMandCDQS}
\end{figure*}

To the best of our knowledge the quantum analogue of the CDS model has not been studied explicitly in the literature.\footnote{It has been studied in an indirect way, since (as we show later) it is equivalent to $f$-routing.} 
We give a definition here which features quantum resources and a quantum secret. 
The CDQS primitive is illustrated in figure \ref{fig:CDQS}. 

\begin{definition}\label{def:CDQS}
    A \textbf{conditional disclosure of quantum secrets (CDQS)} task is defined by a choice of function $f:\{0,1\}^{2n}\rightarrow \{0,1\}$, and a $d_Q$ dimensional Hilbert space $\mathcal{H}_Q$ which holds the secret.
    The task involves inputs $x\in \{0,1\}^{n}$ and system $Q$ given to Alice, and input $y\in \{0,1\}^{n}$ given to Bob.
    Alice sends message system $M_0$ to the referee, and Bob sends message system $M_1$. 
    Label the combined message systems as $M=M_0M_1$.
    Label the quantum channel defined by Alice and Bob's combined actions $\mathbfcal{N}_{Q\rightarrow M}^{xy}$. 
    We put the following two conditions on a CDQS protocol. 
    \begin{itemize}
        \item $\epsilon$\textbf{-correct:} There exists a channel $\mathbfcal{D}^{x,y}_{M\rightarrow Q}$, called the decoder, such that
        \begin{align}
            \forall (x,y)\in X\times Y \,\,\, s.t. \,\, f(x,y)=1,\,\,\, ||\mathbfcal{D}^{x,y}_{M\rightarrow Q}\circ \mathbfcal{N}^{x,y}_{Q\rightarrow M} - \mathbfcal{I}_{Q\rightarrow Q}||_\diamond \leq \epsilon
        \end{align}
        \item $\delta$\textbf{-secure:} There exists a quantum channel $\mathbfcal{S}_{\varnothing \rightarrow M}^{x,y}$, called the simulator, such that
        \begin{align}
            \forall (x,y)\in X\times Y \,\,\, s.t. \,\, f(x,y)=0,\,\,\, ||\mathbfcal{S}_{\varnothing \rightarrow M}^{x,y} \circ \tr_Q - \mathbfcal{N}_{Q\rightarrow M}^{x,y}||_\diamond \leq \delta
        \end{align}
    \end{itemize}
\end{definition}
The notions of $\epsilon$-correctness and $\delta$-security given here mimic the classical ones closely. 
In words, the correctness condition is saying that when $f(x,y)=1$ the referee can reverse the effect of Alice and Bob's actions on the $Q$ system.
The security condition is saying that when $f(x,y)=0$ the system $M$ seen by the referee is close to one that they could have prepared with no access to $Q$. 

In our definition of CDQS, we require a quantum system $Q$ be taken as the secret, and allow the use of quantum resources. 
Another quantum variant of CDS we could have defined would allow quantum resources but restrict to a classical secret. 
We could call this CDQS'.
This variant is in fact equivalent to the above definition.
This follows from our proof below that classical CDS protocols gives quantum CDS protocols, which is easily modified to show a CDQS' gives CDQS with similar resources. 
Then one can observe that a CDQS protocol can be modified to a CDQS' protocol by choosing the secret to be a state in a chosen basis. 
Taken together these observations give that CDQS' and CDQS are equivalent. 

\vspace{0.2cm}
\noindent \textbf{Private simultaneous message passing}
\vspace{0.2cm}

Next we move on to discuss another basic cryptographic primitive of interest in this article, which is private simultaneous message passing. 
This primitive is illustrated in figure \ref{fig:PSM}. 

\begin{definition}\label{def:PSM}
    A \textbf{private simultaneous message (PSM)} task is defined by a choice of function $f:X\times Y\rightarrow Z$. 
    The inputs to the task are $n$ bit strings $x$ and $y$ given to Alice and Bob, respectively. 
    Alice then sends a message $m_0(x,r)$ to the referee, and Bob sends message $m_1(y,r)$. 
    From these inputs, the referee prepares an output bit $z$. 
    We require the task be completed in a way that satisfies the following two properties.
    \begin{itemize}
        \item \textbf{$\epsilon$-correctness:} There exists a decoder $Dec$ such that \begin{align}
            \forall (x,y)\in X\times Y,\, \,\,\,\mathrm{Pr}[Dec(m_0,m_1)=f(x,y)] \geq 1-\epsilon.
        \end{align}
        \item \textbf{$\delta$-security:} There exists a simulator producing a distribution $Sim$ on the random variable $M=M_0M_1$ such that 
        \begin{align}
            \forall (x,y)\in X\times Y,\, \,\,\,||Sim_{M|f(x,y)} - P_{M|xy}||_1\leq \delta.
        \end{align}
        Stated differently, the distribution of the message systems is $\delta$-close to one that depends only on the function value, for every choice of $x,y$. 
    \end{itemize}
\end{definition}

In PSM we can allow the function $f$ to take Boolean or other values. 
For instance we can take $f$ to be natural number valued and defined by a counting problem. 
Another comment is that PSM protocols can be run in parallel, in the sense that $\epsilon$-correct and $\delta$-secure protocols for $f_1(x,y)$ and $f_2(x,y)$ can be run together to give a $2\epsilon$-correct and $2\delta$-secure protocol for the function $f(x,y)=(f_1(x,y),f_2(x,y))$. 
This is straightforward to show from the security definition. 

\vspace{0.2cm}
\noindent \textbf{Private simultaneous quantum message passing (PSQM)}
\vspace{0.2cm}

As with CDS, there is a natural quantum version of PSM. In this case the functionality of the protocol is unchanged, but the allowed resources are now quantum mechanical. 
A PSQM protocol is shown in figure \ref{fig:PSQM}.

\begin{definition}\label{def:PSQM}
    A \textbf{private simultaneous quantum message (PSQM)} task is defined by a choice of function $f:X\times Y\rightarrow Z$. 
    The inputs to the task are $n$ bit strings $x$ and $y$ given to Alice and Bob, respectively, each of which are chosen independently and at random. 
    Alice then sends a quantum message system $M_0$ to the referee, and Bob sends quantum message system $M_1$. 
    From the combined message system $M=M_0M_1$, the referee prepares an output bit $z$.  
    We require the task be completed in a way that satisfies the following two properties.
    \begin{itemize}
        \item \textbf{$\epsilon$-correctness:} There exists a decoding map $\mathbf{V}_{M \rightarrow Z\tilde{M}}$ such that 
        \begin{align}
            \forall (x,y)\in X\times Y, \,\,\,\,\, \left|\left|\tr_{\tilde{M}}(\mathbf{V}_{M \rightarrow Z\tilde{M}} \rho_{M}(x,y) \mathbf{V}_{M \rightarrow Z\tilde{M}}^\dagger ) - \ketbra{f_{xy}}{f_{xy}}_Z\right|\right|_1 \leq \epsilon.
        \end{align}
        where $\rho_M(x,y)$ is the density matrix on $M$ produced on inputs $x,y$.
        \item \textbf{$\delta$-security:} There exists a simulator, which is a quantum channel $\mathbfcal{S}_{Z\rightarrow M}(\cdot)$, such that 
        \begin{align}
            \forall (x,y)\in X\times Y,\,\,\,\,\,\left|\left|\rho_{M}(x,y) - \mathbfcal{S}_{Z\rightarrow M}(\ketbra{f_{xy}}{f_{xy}}_Z)\right|\right|_1 \leq \delta.
        \end{align}
        Stated differently, the state of the message systems is $\delta$-close to one that depends only on the function value, for every choice of input.
    \end{itemize}
\end{definition}

Just like in the classical case, PSQM protocols can be run in parallel with only small relaxations in security and correctness. 

\vspace{0.2cm}
\noindent \textbf{Decomposable randomized encodings}
\vspace{0.2cm}

A related primitive, which we will make briefer use of, is the notion of a decomposable randomized encoding. 
We recall some definitions given in \cite{computation2013randomization}. 

\begin{definition} Let $X,Y,\hat{Y},R$ be finite sets and let $f:X_1\times ... \times X_n \rightarrow Y$. A function $\hat{f}:X\times R\rightarrow \hat{Y}$ is an $\epsilon$-correct and $\delta$-private \textbf{randomized encoding} for $f$ if it satisfies
\begin{itemize}
    \item $\epsilon$-\textbf{correctness:} There exists a function $Dec$ called a decoder such that for every $x\in X$ we have 
    \begin{align}
        \mathrm{Pr}[Dec(\hat{f}(x,r))=f(x)]\geq 1 - \epsilon.
    \end{align}  
    where the probability is over $R$ and any randomness in the decoder $Dec$. 
    \item $\delta$-\textbf{privacy:} There exists a randomized function, called a simulator, producing the random variable $Sim$ such that
    \begin{align}
        ||Sim_{\hat{Y}|Y} - P_{\hat{Y}|X} ||_1\leq \delta.
    \end{align}
\end{itemize}
\end{definition}

\begin{definition}\label{def:DRE} A \textbf{decomposable randomized encoding (DRE)} for a function $f:X_1\times ... \times X_n \rightarrow Y$ is a randomized encoding of $f$ that has the form
\begin{align}
    \hat{f}(x_1,...,x_n;r)= (\hat{f}_1(x_1,r),...,\hat{f}_n(x_n,r))
\end{align}
A DRE is $\epsilon$-correct and $\delta$-secure under the same conditions as a randomized encoding, given above.
\end{definition}
\noindent We will in fact only use that certain randomized encodings are decomposable across a single splitting of the inputs. 
That is we are interested in functions $f:X\times Y\rightarrow Z$ and need the randomized encoding to take the form
\begin{align}
    \hat{f}(x,y;r)= (\hat{f}_1(x,r),\hat{f}_2(y,r))
\end{align}
In this setting we will say $f(x,y)$ has a randomized encoding decomposable across $X\times Y$. 

\vspace{0.2cm}
\noindent \textbf{Non-local computation}
\vspace{0.2cm}

Finally we come to the notion of a non-local computation, which was first studied in the context of cheating strategies for position-verification tasks. 
The general setting is shown in figure \ref{fig:non-localandlocal}.
A non-local computation takes the form shown in figure \ref{fig:non-localcomputation}, with the goal being to simulate the action of a local unitary (figure \ref{fig:local}).

We will not give a formal definition of a fully general NLQC here, but instead focus on two special cases. 
The first, $f$-routing, was introduced in \cite{kent2011quantum} and studied further in \cite{buhrman2013garden}. 
It has been especially well studied in the non-local computation literature because it is of interest in developing practical position-verification schemes. 
We will also see that it is closely related to the CDQS primitive.\footnote{Our definition here gives a particular notion of an $\epsilon$-correct $f$-routing scheme, which requires the protocol route an arbitrary quantum state correctly. Other definitions \cite{bluhm2022single} require correct action on only the maximally entangled state. For inputs of a fixed size these are equivalent.} 

\begin{definition}\label{def:frouting}
    A \textbf{$f$-routing} task is defined by a choice of Boolean function $f:\{ 0,1\}^{2n}\rightarrow \{0,1\}$, and a $d$ dimensional Hilbert space $\mathcal{H}_Q$.
    Inputs $x\in \{0,1\}^{n}$ and system $Q$ are given to Alice, and input $y\in \{0,1\}^{n}$ is given to Bob.
    Alice and Bob exchange one round of communication, with the combined systems received or kept by Bob labelled $M$ and the systems received or kept by Alice labelled $M'$.
    Label the combined actions of Alice and Bob in the first round as $\mathbfcal{N}^{x,y}_{Q\rightarrow MM'}$. 
    The $f$-routing task is completed $\epsilon$-correctly if there exists a channel $\mathbfcal{D}^{x,y}_{M\rightarrow Q}$ such that,
    \begin{align}
        \forall (x,y)\in X\times Y \,\,\, s.t. \,\, f(x,y)=1,\,\,\, ||\mathbfcal{D}^{x,y}_{M\rightarrow Q} \circ\tr_{M'} \circ\mathbfcal{N}^{x,y}_{Q\rightarrow MM'} -\mathbfcal{I}_{Q\rightarrow Q}||_\diamond \leq \epsilon
    \end{align}
    and there exists a channel $\mathbfcal{D}^{x,y}_{M'\rightarrow Q}$ such that
    \begin{align}
        \forall (x,y)\in X\times Y \,\,\, s.t. \,\, f(x,y)=0,\,\,\,||\mathbfcal{D}^{x,y}_{M'\rightarrow Q} \circ\tr_{M} \circ\mathbfcal{N}^{x,y}_{Q\rightarrow MM'} -\mathbfcal{I}_{Q\rightarrow Q}||_\diamond \leq \epsilon
    \end{align}
    In words, Bob can recover $Q$ if $f(x,y)=1$ and Alice can recover $Q$ if $f(x,y)=0$. 
\end{definition}

The second special case we study is coherent function evaluation. 
We introduce this as the special case of NLQC that implies the PSQM primitive, as we show below. 
As well, it is similar to non-local computations studied in \cite{junge2021geometry}, which used Banach space techniques to study lower bounds on quantum resources in these non-local computations. 

\begin{definition}\label{def:CFE}
    A \textbf{coherent function evaluation (CFE)} task is defined by a choice of Boolean function $f:\{0,1\}^{2n}\rightarrow \{0,1\}$.
    The task is to implement the isometry
    \begin{align}
        \mathbf{V}_f = \sum_{xy} \ket{xy}_{Z'} \ket{f_{xy}}_{Z} \bra{x}_X\bra{y}_Y
    \end{align}
    in the non-local form of figure \ref{fig:non-localcomputation}. 
    We say a CFE protocol is $\epsilon$-correct if the diamond norm distance between $\mathbf{V}_f$ and the implemented channel is not larger than $\epsilon$.  
\end{definition}

\vspace{0.2cm}
\noindent \textbf{Secret sharing}
\vspace{0.2cm}

An important tool throughout cryptography, and in particular in our context, is the notion of a secret sharing scheme.
We introduce this next. 

\begin{definition} \label{def:SS}
    A \textbf{secret sharing scheme} $\mathbf{S}$ is a map from a domain $K$ and randomness $R$ to variables $S_1,...,S_n$, here called shares.
    Let $A$ be a subset of the $S_i$, $\mathcal{S}_A$ the distribution on the shares $A$, and $\mathbf{A}$ a set of subsets of the $S_i$. 
    Then a scheme $\mathbf{S}$ realizes access structure $\mathbf{A}$ with \textbf{$\epsilon$-correctness} if, for each subset of shares $A\in \mathbf{A}$ there exists a decoding map $D_A:A\rightarrow K$ such that
    \begin{align}
        \forall s\in K, \,\,\,\,\, \mathrm{Pr}[D_A(S_A)=s]\geq 1-\epsilon. 
    \end{align}
    A scheme $\mathbf{S}$ is \textbf{$\delta$-secure} if, whenever $U\notin \mathbf{A}$, there exists a map producing a distribution $\text{Sim}$ on $U$ such that
    \begin{align}
        ||\text{Sim}_{U} - \mathcal{S}_{U|K}||_1 \leq \delta
    \end{align}
    If $\epsilon=\delta=0$ we say that the scheme $\mathbf{S}$ is perfect. 
\end{definition}
The access structure of a secret scheme can be specified as a set of subsets of shares, as in the above definition, or equivalently in terms of an \textbf{indicator function}. 
The indicator function is defined by
\begin{align}
    f_I(x) = \begin{cases}
		1 & \text{if} \,\,\,\,\{S_i:x_i=1\} \in \mathcal{A}\\
            0 & \text{otherwise}
		 \end{cases}
\end{align}
We can observe that if $A\in \mathcal{A}$ then necessarily $A\cup S_{i}\in\mathcal{A}$. 
This follows because if we can reconstruct the secret from $A$, we can also reconstruct it from a larger set.
This means that valid indicator functions will always be monotone. 

\vspace{0.2cm}
\noindent \textbf{The garden hose game}
\vspace{0.2cm}

The garden hose game \cite{buhrman2013garden} is a model of communication complexity defined, informally, as follows. 
Alice and Bob are neighbours, and wish to compute a function $f(x,y)$, where Alice holds the input $x$ and Bob the input $y$. 
They have a set of $m$ pipes that run through their fence and connect the two yards. 
Alice has a tap, which she can connect to any of the pipe openings on her side of the fence. 
Alice and Bob additionally have hoses, which they can use to connect ends of pipes on the same side of the fence. 
Their strategy is to choose how to connect the tap to the pipes, and connect pipes to each other with hoses, in a way that depends on their respective inputs. 
Then, Alice turns on the tap. 
Alice and Bob win the garden hose game if the water spills on Alice's side of the fence when $f(x,y)=0$, and on Bob's side of the fence when $f(x,y)=1$. 
For a formal definition of the garden-hose game, we refer the reader to \cite{buhrman2013garden}. 

The garden hose game gives an interesting notion of the communication complexity of a function, which we formalize next. 

\begin{definition}\label{def:gh}
The \textbf{garden hose complexity} of a function $f:\{0,1\}^n\times \{0,1 \}^n\rightarrow \{0,1\}$ is the minimal number of pipes needed to complete the garden hose game for the function $f(x,y)$ deterministically. 
\end{definition}

All functions can be computed in the garden hose game. 
To see why, observe that for any $f(x,y)$ Alice and Bob can carry out the following strategy. 
They prepare $2^{n+1}$ pipes, which we label as $\{p_i,p_i'\}_{i=1}^{n}$. 
Upon receiving input $x$, Alice connects her tap to pipe $p_x$. 
Bob connects pipe $p_i$ to $p_i'$ whenever $f(i,y)=0$, and leaves it open otherwise. 
Upon turning on the tap then, water flows through pipe $p_x$, then back to Alice if $f(x,y)=0$ and spills on the right otherwise, as needed. 
A sightly smarter strategy lowers the worst case garden hose complexity to $2^n+1$. 
See \cite{buhrman2013garden}.

\vspace{0.2cm}
\noindent \textbf{Other related primitives}
\vspace{0.2cm}
 
Each of the primitives discussed above is in turn related to others in various ways. 
Reviewing these further connections is outside the scope of this article. 
Instead, we have included in our discussion only new connections among primitives, or primitives for which we have found the connection to NLQC gives a new result on NLQC, or for which NLQC implies a new result on the primitive.
We briefly mention however some settings with natural relationships to the ones discussed here; our list and references are not exhaustive.  
CDS and PSM are related to zero-knowledge proofs \cite{applebaum2017private}, secret sharing \cite{applebaum2020better}, communication complexity \cite{applebaum2021placing},
private information retrieval \cite{ishai1997private}, and secure multiparty computation \cite{ishai1997private}. 
A useful review of these primitives and the broader context of information theoretic cryptography is given in \cite{BIUschool}. 
Quantum secret sharing was related to $f$-routing in \cite{cree2022code}. 
All of these connections may be interesting to revisit in the quantum setting, and in light of the connection to non-local computation and position-verification. 

\subsection{Existing relations among primitives}\label{sec:existingrelationships}

\vspace{0.2cm}
\noindent \textbf{SS gives CDS}
\vspace{0.2cm}

In \cite{GERTNER2000592}, the authors upper bound the randomness complexity of a CDS scheme in terms of the size of a secret sharing scheme whose access structure is related to $f$. 
We recall their result next, narrowing their result to the two player case for simplicity.  
\begin{theorem}\label{thm:SSgivesCDS} Let $f_M:\{0,1\}^m\times \{0,1\}^m \rightarrow \{0,1\}$ be a monotone Boolean function and let $f:\{0,1\}^n\times \{0,1\}^n \rightarrow \{0,1\}$ be a projection of $f_M$, that is $f(x,y) = f_M(g_1(x),g_2(y))$.
Let $\mathbf{S}$ be a perfect secret sharing scheme realizing the access structure $f_M$, in which the total share size is $c$, and let $s$ denote a secret (from the domain of $\mathbf{S}$) which is known to all players.
Then there exists a CDS protocol for disclosing $s$ subject to the condition $f$ with randomness $c$, and a (perhaps different) protocol with communication complexity bounded above by $c$. 
\end{theorem}
The protocol which establishes this theorem is, heuristically, the following. 
We start by illustrating the case where $f=f_M$ is already a monotone function, and so can be realized as the indicator function of some secret sharing scheme $\mathbf{S}$.
Then the protocol is as follows.
Without loss of generality take Alice and Bob to both hold the secret $s$ (see Remark \ref{remark:onesidedCDS}). 
To carry out the protocol, both parties prepare a secret sharing scheme $\mathbf{S}$ which has indicator function $f_M$, using their shared randomness as the randomness $R$ needed to prepare the scheme. 
Then, Alice sends those shares $S_i$ to the referee for which $x_i=1$, and Bob sends those shares $S_{i+n}$ for which $y_i=1$. 
Then if $f_M(x,y)=1$, following this local rule they will have collectively sent an authorized set of shares and the referee can reconstruct the secret $s$. 
If $f_M(x,y)=0$, they will have sent an unauthorized set of shares and the referee cannot learn the secret. 
To extend this to non-monotone functions, Alice and Bob first locally compute $g_1$ and $g_2$ respectively, and then perform the same secret sharing protocol now with bits of $g_1(x)$ or $g_2(y)$ controlling which shares are sent to the referee. 
Notice that the communication complexity is at most the total size of the shares of the secret sharing scheme. 

To see the protocol that gives an upper bound for the randomness complexity\footnote{This protocol is not given in \cite{GERTNER2000592}, but is a straightforward extension of their idea.}, we now have only Alice prepare the shares of the secret sharing scheme. 
For shares $i\leq n$, she sends share $S_i$ if $x_i=1$ as before. 
For shares $i>n$, she sends $S_i\oplus r_i$, where the XOR is taken bitwise with a random string $r_i$ of length $|S_i|$. 
Bob then sends $r_i$ iff $y_i=1$. 
Notice that the randomness complexity is now at most $\sum_i r_i \leq \sum_i |S_i|$, which is just the size of the scheme. 
The communication complexity is now somewhat larger, but is bounded by twice the size. 

We can also generalize the above theorem to the case of approximate secret sharing schemes. 
In particular, if we use an approximate secret sharing scheme in the second of the protocols above we find that an $\epsilon$-correct and $\delta$-secure secret sharing scheme of size $c$ for an indicator function $f_I$ leads to an $\epsilon$-correct and $\delta$-secure CDS for the same function, using randomness complexity $c$.
A similar observation holds for the protocol bounding the communication complexity. 
We collect these observations as the following remark. 

\begin{remark}\label{thm:robustCDSfromSS}
Let $f_M:\{0,1\}^m\times \{0,1\}^m \rightarrow \{0,1\}$ be a monotone Boolean function and let $f:\{0,1\}^n\times \{0,1\}^n \rightarrow \{0,1\}$ be a projection of $f_M$, that is $f(x,y) = f_M(g_1(x),g_2(y))$.
Let $\mathbf{S}$ be an $\epsilon$-correct and $\delta$-secure secret sharing scheme realizing the access structure $f_M$, in which the total share size is $c$, and let $s$ denote a secret (from the domain of $\mathbf{S}$) which is known to all players.
Then there exists an $\epsilon$-correct and $\delta$-secure CDS protocol disclosing $s$ subject to the condition $f$ with randomness $c$, and a (perhaps different) $\epsilon$-correct and $\delta$-secure protocol with communication complexity bounded above by $c$. 
\end{remark}

\vspace{0.2cm}
\noindent \textbf{DRE gives PSM}
\vspace{0.2cm}

See for example \cite{computation2013randomization} for the connection between DRE and PSM. 
We give a robust version of this connection as the next theorem. 

\begin{theorem} \label{thm:DREgivesPSM}
Suppose that $f:X\times Y\rightarrow Z$ has an $\epsilon$-correct and $\delta$-secure decomposable randomized encoding using $n_R$ bits of randomness, and $n_M$ message bits. 
Then there is an $\epsilon$-correct and $\delta$-secure PSM protocol for $f$ that uses the same amount of randomness and message bits. 
\end{theorem}
\begin{proof}
\,Let the DRE for $f$ be
\begin{align}
    \hat{f}(x,y;r) = (\hat{f}_X(x,r),\hat{f}_Y(y,r)) 
\end{align}
To implement the PSM protocol, Alice prepares $\hat{f}_X(x,r)$ and sends this to the referee, while Bob prepares $\hat{f}_Y(y,r)$ and sends this to the referee. 
The referee then uses the decoder for the DRE to determine $f(x)$. 
Noticing that the conditions on the DRE and PSM are in fact exactly the same under these identifications, we have that the PSM is also $\epsilon$-correct and $\delta$-secure. 
\end{proof}

Notice that a PSM for $f$ also gives a randomized encoding for the function $f$, albeit one that is decomposable across a particular splitting of the input bits into $X\times Y$, and not necessarily decomposable bitwise, as required in the definition of a DRE. 

\vspace{0.2cm}
\noindent \textbf{PSM gives CDS}
\vspace{0.2cm}

Next, we relate the PSM and CDS primitives. 
See for example \cite{GERTNER2000592,applebaum2017private}.

\begin{theorem}\label{thm:PSMgivesCDS}
Suppose that an $\epsilon$-correct and $\delta$-private PSM protocol exists for $f(x,y)$ using messages of at most $n_M$ bits and no more than $n_E$ shared random bits. 
Then a CDS protocol using $n_M+1$ bits of message and $n_E$ random bits exists which is $\epsilon$-correct and $O(\delta \log d_R)$ private, and hides one bit. 
\end{theorem}
\begin{proof}\,
We wish to carry out the CDS task using the given PSM protocol. 
First, we note that by adding one bit of randomness we can assume $s$ is held by both Alice and Bob. 
This is because of remark \ref{remark:onesidedCDS}. 

Next, we show that given the PSM protocol for $f$ there is a similarly efficient PSM for the function $f(x,y)\wedge s$, with $s$ held on both sides. 
To show this, first consider the case where $f(x,y)$ is a constant function. 
Then Alice and Bob can follow a fixed strategy (reveal $s$ or not) and we are done. 
Thus we assume $f(x,y)$ is non-constant, and choose any input values for which it is $0$ and label them $(x_*,y_*)$. 
Run the PSM on inputs $x'=sx+ (1-s) x_*$ and $y'=sy + (1-s) y_*$.
Then notice that $f(x',y')=f(x,y)\wedge s$. 

To see $\epsilon$-correctness, we have the referee output the outcome of the modified PSM protocol as their guess for the secret $s$. 
Then their success probability conditioned on $f(x,y)=1$ is exactly $1-\epsilon$, so the CDS protocol is $1-\epsilon$ correct. 

Next consider security. Let the distribution of values of $f(x,y)$ be $F$, the distribution of values of $f(x',y')$ be $F'$, and the distribution of $x'$ and $y'$ be $X'$ and $Y'$ respectively. 
Security of the original PSM protocol implies
\begin{align}
    ||Sim_{M|F'} - P_{M|X'Y'}||_1\leq \delta.
\end{align}
Then notice that because $X'Y'$ are determined by $XYS$, we have $P_{M|X'Y'}=P_{M|XYS}$.
Next, restrict to the distributions where $f(x,y)=0$, leading to
\begin{align}
    ||Sim_{M|F'=0} - P_{M|XYS}||_1\leq \delta
\end{align}
which is $\delta$ security of the CDS. 
\end{proof}

\vspace{0.2cm}
\noindent \textbf{PSM gives PSQM}
\vspace{0.2cm}

Next, we prove that a protocol for PSM also gives a protocol for PSQM.
This might seem trivial, since the quantum resources available in the PSQM can simulate the classical resources used in the PSM, but establishing security requires we show the classical security definition is strong enough to enforce the quantum security definition. 
As far as we are aware this is not written in the literature (but see \cite{kawachi2021communication} for the introduction of PSQM), but is straightforward enough we include it in this section. 

\begin{theorem}\label{thm:PSMgivesPSQM}
    Suppose we have a PSM protocol which is $\epsilon$-correct and $\delta$-secure. 
    Then we can construct a PSQM protocol which is $2\sqrt{\epsilon}$ correct and $\delta$-secure. 
\end{theorem}
\begin{proof}\,
Correctness of the PSM protocol implies that there exists a decoder $Dec(m_0,m_1)$ such that
\begin{align}
    \forall (x,y)\in X\times Y\,\,\,\, \mathrm{Pr}[Dec(m_0,m_1) = f(x,y)]\geq 1-\epsilon
\end{align}
where the probability is over choices of the random string $r$. 
In quantum notation, we have that the message system is described by the density matrix
\begin{align}
    \rho_M(x,y) = \sum_{m} p(m|x,y) \ketbra{m}{m}
\end{align}
and can write the output of the decoder as
\begin{align}
    \mathbfcal{D}_{M\rightarrow Z}(\rho_M(x,y)) = \sum_{m} p(m|x,y) \ketbra{D(m)}{D(m)}.
\end{align}
Then notice that
\begin{align}
    F( \mathbfcal{D}_{M\rightarrow Z}(\rho_M(x,y)),\ket{f_{xy}}) = \sum_m p(m|x,y) |\braket{D(m)}{f_{xy}}|^2 \geq 1-\epsilon
\end{align}
where the last line follows because we see the fidelity is exactly the guessing probability, which is bounded from below by the classical correctness definition. 
Using the Fuchs van de Graff inequalities, we get that 
\begin{align}
    ||\mathbfcal{D}_{M\rightarrow Z}(\rho_M(x,y)) - \ketbra{f_{xy}}{f_{xy}}||_1\leq 2\sqrt{\epsilon}
\end{align}
as needed. 

Next recall security of the PSM means that there exists a simulator which takes in $f(x,y)$ and produces output distribution $Sim$ on the message system such that
\begin{align}
    \forall (x,y)\in X\times Y,\, \,\,\,||Sim_{M|f(x,y)} - P_{M|xy}||_1 \leq \delta.
\end{align}
To get security of the PSQM, we need to upgrade this simulator to a quantum channel. 
In particular if the simulator is defined by the conditional probability distribution $p(m|f)$, define the Kraus operators
\begin{align}
    S_{m,f} = \sqrt{p(m|f)} \ketbra{m}{f}. 
\end{align}
Calling the corresponding simulator channel $\mathbfcal{S}$, we have that
\begin{align}
    ||\mathbfcal{S}(\ketbra{f_{xy}}{f_{xy}}) - \rho_M(x,y)||_1= ||Sim_{M|f(x,y)} - P_{M|xy}||_1\leq \delta
\end{align}
so we have exactly $\delta$ security of the PSQM. 
\end{proof}

\vspace{0.2cm}
\noindent \textbf{GH gives $f$-routing}
\vspace{0.2cm}

In \cite{buhrman2013garden}, the following statement is shown. 
\begin{theorem}\label{thm:GHtofR}
    The number of EPR pairs needed to implement a $f$-routing protocol for a function $f$ is upper bounded by the garden hose complexity of $f$.
\end{theorem}

We won't reproduce a careful proof of this, but it is easy to see: each pipe in the garden hose protocol is replaced with an EPR pair in the f-routing strategy. 
Connecting pipes corresponds to measuring pairs of systems in the Bell basis.
Doing so, the input system $Q$ will end up recorded into the Hilbert space corresponding to spilling end of one of the pipes. 
Pauli corrections appear on this state, but the one round of communication in the $f$-routing strategy can be used to communicate all the measurement outcomes and then undo the corresponding corrections. 

\section{New relations among primitives}\label{eq:newrelationships}

This section begins our study of the relationships among the cryptographic primitives introduced in section \ref{sec:primitivedefinitions}. 

\subsection{Garden hose strategies give CDS}

We point out that the garden hose game defines strategies for CDS.
\begin{theorem}\label{thm:GHtoCDS}
    The garden hose complexity of a function $f(x,y)$ upper bounds the CDS cost, 
    \begin{align}
        CDS(f) \leq GH(f)
    \end{align}
\end{theorem}
\begin{proof}\,
To show this, we construct a CDS protocol given a garden-hose protocol that uses a number of shared random bits equal to the number of pipes in the garden hose protocol. 

Label the set of pipes used in the garden hose game $p_i$ with the tap labelled $p_0$, the connections on Alice's side by $C_x=\{(p_i,p_j) \}$, and the connections on Bob's side by $C_y=\{(p_i,p_j)\}$.
Note that because no pipe can be connected to two hoses, each $p_i$ appears in $C_x$ at most once, and in $C_y$ at most once. 
Correctness of the garden hose protocol means that for all $(x,y)$, there is a path from the tap to the side labelled by $f(x,y)$. 

To turn this into a CDS protocol, we proceed as follows. 
Each pipe $p_i$, $i>0$, becomes a shared random bit held by Alice and Bob. 
The secret $s$ corresponds to the tap $p_0$. 
For each connection in $C_x$, say $(p_i,p_j)$, Alice computes $c_{ij}=p_i\oplus p_j$ and sends this to the referee. 
Bob does the same for each connection in $C_y$. 
Finally, Bob sends each shared random bit $p_k$ not appearing in any connection in $C_y$ to the referee. 
In contrast, Alice's unused random bits are kept hidden from the referee. 

To see why this is correct and secure, consider the chain of connection bits $c_{i_ki_{k+1}}=p_{i_k}\oplus p_{i_{k+1}}$, where $p_{i_0}=s$ is the secret. 
If the chain is of length 0, this corresponds to an unconnected tap in the garden hose picture, so that $f(x,y)=0$ and the water spills on the left. 
In the CDS protocol, the secret, being an un-XOR'd bit, is not sent to the referee, so that the referee cannot learn the secret, as needed. 
Now suppose the chain has length $>1$. 
Then $c_{i_0i_1}=s\oplus p_{i_1}$ is sent to the referee, and no other bits which are computed from $s$ are sent, so that the referee learns $s$ if and only if they learn $p_{i_1}$. 
Continuing in this way down the chain of connection bits, we see that the referee learns $s$ if and only if they learn $p_{i_m}$, the final random bit (corresponding to the final pipe in the waters path). 
But then $p_{i_m}$ is not used to compute any other bits (by virtue of being at the end of the chain), and is sent if and only if it is unused on the right.
But it is unused on the right if and only if the corresponding pipe spills on the right, which by our assumption of correctness of the garden hose strategy is if and only if $f(x,y)=1$
\end{proof}

\subsection{Classical CDS gives quantum CDS}

In this section we observe that a classical CDS scheme immediately gives a quantum CDS scheme, via a use of the one-time pad. 

\begin{theorem}\label{thm:CDStoCDQS}
    An $\epsilon$-correct and $\delta$-secure CDS protocol hiding $2n$ bits and using $n_M$ bits of message and $n_E$ bits of randomness gives a CDQS protocol which hides $n$ qubits, is $2\sqrt{\epsilon}$ correct and $\delta$-secure using $n_M$ classical bits of message plus $n$ qubits of message, and $n_E$ classical bits of randomness.
\end{theorem}
\begin{proof}\,
Let the quantum system to be hidden in the CDQS be labelled $Q$. 
The basic idea is to use the CDS protocol to hide the key of a one-time pad applied to the system $Q$.
The encoded system $Q$ is sent to the referee.
The one-time pad key, call it $s$, consists of $2\log d_Q$ bits, which we choose independently and at random and hide in the CDS. 
The channel applied by Alice and Bob's combined actions is then
\begin{align}\label{eq:CDQSchannel}
    \mathbfcal{N}^{xy}_{Q\rightarrow QM}(\cdot) = \frac{1}{2^{|s|}} \sum_{m,s} P_Q^s(\cdot)P_Q^s \otimes p_{m|xys}\ketbra{m}{m}_M
\end{align}
We first study correctness.
To do this, we recall that correctness of the classical CDS guarantees the existence of a decoder which produces an outcome which is equal to the secret value with probability $1-\epsilon$. 
In quantum notation, we can describe this channel as
\begin{align}
    \text{Dec}^{xy}_{M\rightarrow S}(\cdot) = \sum_{m,s'} p_{s'|mxy} \ket{s'}_S\bra{m}_M \cdot \ket{m}_M\bra{s'}_S.
\end{align}
The correctness condition for CDS states that, for $(x,y)\in f^{-1}(1)$ this produces a guess $s'$ which agrees with the secret $s$, or more precisely,
\begin{align}
    F\left(\text{Dec}^{xy}(\sum_m p_{m|sxy}\ketbra{m}{m}),\ketbra{s}{s}\right) \geq 1-\epsilon.
\end{align}
Relating this to the trace distance via the Fuchs van de Graff inequalities, this becomes, 
\begin{align}\label{eq:CDSTDcorrectness}
    \sum_{s'}\left(\sum_{m}p_{s'|mxy}p_{m|sxy} - \delta_{s'|s} \right)\leq 2\sqrt{\epsilon},
\end{align}
where $\delta_{s|s'}=1$ if $s=s'$ and is zero otherwise.
We will use this statement in establishing correctness of the CDQS. 

Define the decoding channel for the CDQS by combining the classical decoder with a conditional application of $P^{s'}_Q$, then a trace over the register $S$ holding the secret, so that our decoder is
\begin{align}
    \mathbfcal{D}^{xy}_{QM\rightarrow Q}(\cdot) = \sum_{m,s'} p_{s'|mxy} P^{s'}_Q \otimes\bra{m}_M \cdot P^{s'}_Q\otimes \ket{m}_M
\end{align}
We need to bound the diamond norm $||\mathbfcal{D}^{xy}_{QM\rightarrow Q} \circ \mathbfcal{N}^{xy}_{Q\rightarrow M'} - \mathbfcal{I}_{Q\rightarrow Q}||_\diamond$ from above. 
From the definition of the diamond norm and the channels $\mathbfcal{D}^{xy}_{QM\rightarrow Q},\mathbfcal{N}^{xy}_{Q\rightarrow M'}$, this is
\begin{align}
    ||\mathbfcal{D}_{QM\rightarrow Q} &\circ \mathbfcal{N}^{xy}_{Q\rightarrow M'} - \mathbfcal{I}_{Q\rightarrow Q}||_\diamond \nonumber \\
    &= \sup_n \max_{\Psi_{R_nQ}} ||\frac{1}{2^{|s|}}\sum_{m,s,s'} p_{s'|mxy}p_{m|sxy}P^{s+s'}_Q \Psi_{R_nQ}P^{s+s'}_Q - \Psi_{R_nQ}||_1 \nonumber \\
    &= \sup_n \max_{\Psi_{R_nQ}} ||\frac{1}{2^{|s|}}\sum_{m,s,s'} p_{s'|mxy}p_{m|sxy}P^{s+s'}_Q \Psi_{R_nQ}P^{s+s'}_Q - \frac{1}{2^{|s|}}\sum_{s,s'} \delta_{s|s'}P^{s+s'}_Q\Psi_{R_nQ}P^{s+s'}_Q||_1 \nonumber \\
    &=  \frac{1}{2^{|s|}}\sum_{s,s'}(\sum_m p_{s'|mxy}p_{m|sxy} -\delta_{s'|s}) \sup_n \max_{\Psi_{R_nQ}} || P^{s+s'}_Q \Psi_{R_nQ}P^{s+s'}_Q||_1 \nonumber \\
    &=\frac{1}{2^{|s|}}\sum_{s,s'}(\sum_m p_{s'|mxy}p_{m|sxy} -\delta_{s'|s}) \nonumber \\
    &\leq 2\sqrt{\epsilon}
\end{align}
where we used equation \ref{eq:CDSTDcorrectness} in the last line, which recall held for all $(x,y)\in f^{-1}(1)$. 

To establish security of the CDQS, we define the simulator channel as\footnote{Notice a potential confusion around the notation here: $M$ is the message system of the CDS, $M'=MQ$ is the message system of the CDQS.}
\begin{align}
    \mathbfcal{S}_{\varnothing \rightarrow MQ}^{xy} = \frac{\mathcal{I}_Q}{d_Q}\otimes \sum_m \text{Sim}_{M|xy}\ketbra{m}{m}_M
\end{align}
We need to show $\mathbfcal{S}_{\varnothing \rightarrow MQ}^{xy}\circ \tr_Q$ is close to the channel \ref{eq:CDQSchannel} in diamond norm for all $(x,y)\in f^{-1}(0)$. 
This follows from security of the CDQS and a simple calculation. 
Start with the definition of the diamond norm,
\begin{align}
    ||\mathbfcal{S}_{\varnothing \rightarrow MQ}^{xy}&\circ \tr_Q - \mathbfcal{N}_{Q\rightarrow QM}||_\diamond \nonumber \\
    &= \sup_{n}\max_{\Psi_{R_nQ}}||\mathbfcal{S}_{\varnothing \rightarrow MQ}^{xy}\circ \tr_Q(\Psi_{R_nQ}) - \mathbfcal{N}^{xy}_{Q\rightarrow QM}(\Psi_{R_nQ})||_1 \nonumber \\
    &=\sup_{n}\max_{\Psi_{R_nQ}}||\Psi_{R_n}\otimes \frac{\mathcal{I}_Q}{d_Q}\otimes \sum_m \text{Sim}_{m|xy}\ketbra{m}{m}_M - \frac{1}{2^{|s|}}\sum_{m,s}P^s_Q\Psi_{R_nQ}P^s_Q\otimes p_{m|xys}\ketbra{m}{m}_M||_1 \nonumber \\
    &= \sup_{n}\max_{\Psi_{R_nQ}}||\frac{1}{2^{|s|}}\sum_{m,s}P^s_Q\Psi_{R_nQ}P^s_Q\otimes \text{Sim}_{m|xy}\ketbra{m}{m}_M-\frac{1}{2^{|s|}}\sum_{m,s}P^s_Q\Psi_{R_nQ}P^s_Q\otimes p_{m|xys}\ketbra{m}{m}_M||_1 \nonumber \\
\end{align}
where we used that
\begin{align}
    \Psi_{R_n}\otimes \frac{\mathcal{I}}{d_Q}=\frac{1}{2^{|s|}}\sum_{s}P^s_Q\Psi_{R_nQ}P^s_Q.
\end{align}
To bound our remaining expression, we take the sum over $s$ out of the trace distance and find
\begin{align}
    ||\mathbfcal{S}_{\varnothing \rightarrow MQ}^{xy}&\circ \tr_Q - \mathbfcal{N}_{Q\rightarrow QM}||_\diamond \nonumber \\
    &= \frac{1}{2^{|s|}} \sum_s ||(P^s_Q \Psi_{R_nQ}P^{s}_Q)\otimes  \left(\sum_m \text{Sim}_{m|xy}\ketbra{m}{m}_M - \sum_m p_{m|xys}\ketbra{m}{m}_M\right) ||_1\nonumber \\
    &= \frac{1}{2^{|s|}} \sum_s ||\text{Sim}_{M|xys}-p_{M|xys} ||_1 \nonumber \\
    &\leq \delta
\end{align}
where the last inequality is coming from security of the classical CDS. 
\end{proof}

\subsection{Equivalence of \texorpdfstring{$f$}{TEXT}-routing and CDQS}

Our main claim of this section is that the CDQS and $f$-routing scenarios are equivalent, in that a protocol for one induces a protocol for the other using similar resources. 
The basic idea underlying the equivalence, and labelling of the various subsystems used in the proof, is illustrated in figure \ref{fig:NLQCandCDQS}.

\begin{figure*}
    \centering
    \begin{subfigure}{0.45\textwidth}
    \centering
    \begin{tikzpicture}[scale=0.4]
    
    \draw[thick] (-5,-5) -- (-5,-3) -- (-3,-3) -- (-3,-5) -- (-5,-5);
    
    \draw[thick] (5,-5) -- (5,-3) -- (3,-3) -- (3,-5) -- (5,-5);
    
    \draw[thick] (5,5) -- (5,3) -- (3,3) -- (3,5) -- (5,5);
    
    \draw[thick] (4.5,-3) -- (4.5,3);
    
    \draw[thick] (-3.5,-3) to [out=90,in=-90] (3.5,3);
    
    \draw[thick] (-3.5,-5) to [out=-90,in=-90] (3.5,-5);
    \draw[black] plot [mark=*, mark size=3] coordinates{(0,-7.05)};
    
    \draw[thick] (-4.5,-6) -- (-4.5,-5);
    \node[below] at (-4.5,-6) {$x,Q$};
    
    \draw[thick] (4.5,-6) -- (4.5,-5);
    \node[below] at (4.5,-6) {$y$};

    \node[left] at (-1,-2) {$M_0$};
    \node[right] at (4.5,0) {$M_1$};
    
    \draw[thick] (4,5) -- (4,6);
    \node[above] at (4,6) {$Q$};

    \node at (-4,-4) {$\mathbf{\mathbfcal{N}}^L$};
    \node at (4,-4) {$\mathbf{\mathbfcal{N}}^R$};
    \node at (4,4) {$\mathbf{W}^R$};
    
    \end{tikzpicture}
    \caption{}
    \label{fig:CDQSagain}
    \end{subfigure}
    \hfill
    \begin{subfigure}{0.45\textwidth}
    \centering
    \begin{tikzpicture}[scale=0.4]
    
    \draw[thick] (-5,-5) -- (-5,-3) -- (-3,-3) -- (-3,-5) -- (-5,-5);
    \node at (-4,-4) {$\mathbf{V}^L$};
    
    \draw[thick] (5,-5) -- (5,-3) -- (3,-3) -- (3,-5) -- (5,-5);
    \node at (4,-4) {$\mathbf{V}^R$};
    
    \draw[thick] (5,5) -- (5,3) -- (3,3) -- (3,5) -- (5,5);
    \node at (4,4) {$\mathbf{W}^R$};
    
    \draw[thick] (-5,5) -- (-5,3) -- (-3,3) -- (-3,5) -- (-5,5);
    \node at (-4,4) {$\mathbf{W}^L$};
    
    \draw[thick] (-4.5,-3) -- (-4.5,3);
    
    \draw[thick] (4.5,-3) -- (4.5,3);
    
    \draw[thick] (-3.5,-3) to [out=90,in=-90] (3.5,3);
    
    \draw[thick] (3.5,-3) to [out=90,in=-90] (-3.5,3);
    
    \draw[thick] (-3.5,-5) to [out=-90,in=-90] (3.5,-5);
    \draw[black] plot [mark=*, mark size=3] coordinates{(0,-7.05)};
    
    \draw[thick] (-4.5,-6) -- (-4.5,-5);
    \draw[thick] (4.5,-6) -- (4.5,-5);
    
    \draw[thick] (4.5,5) -- (4.5,6);
    \draw[thick] (-4.5,5) -- (-4.5,6);
    
    \draw[thick] (3.5,5) -- (3.5,6);
    \draw[thick] (-3.5,5) -- (-3.5,6);

    \node[left] at (-1,-2) {$M_0$};
    \node[right] at (4.5,0) {$M_1$};

    \node[right] at (1,-2) {$M_1'$};
    \node[left] at (-4.5,0) {$M_0'$};
    
    \end{tikzpicture}
    \caption{}
    \label{fig:f-routing}
    \end{subfigure}
    \caption{Corresponding CDQS (left) and $f$-routing (right) protocols. To define the CDQS protocol from the $f$-routing protocol, we have Alice and Bob trace out systems $M_0'$ and $M_1'$. Systems $M_0$ and $M_1$ are sent to the referee rather than to Bob. To define the $f$-routing protocol from the CDQS, purify the local channels $\mathbf{\mathbfcal{N}}^L$ and $\mathbf{\mathbfcal{N}}^R$ to isometries $\mathbf{V}^L$ and $\mathbf{V}^R$. Send the original outputs of the channel to Bob on the right, and the purifying systems to Alice on the left. We adopt the notation $M=M_0M_1$ and $M'=M_0'M_1'$.}
    \label{fig:NLQCandCDQS}
\end{figure*}
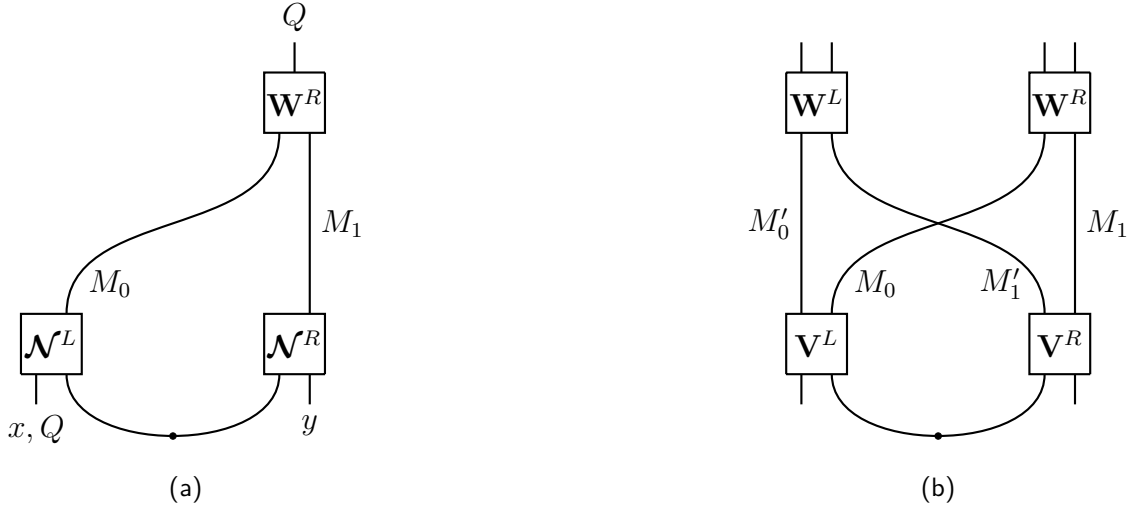

\begin{theorem}\label{thm:CDQSandfRouting}
    an $\epsilon$-correct $f$-routing protocol that routes $n$ qubits implies the existence of an $\epsilon$-correct and $\delta=2\sqrt{\epsilon}$-secure CDQS protocol that hides $n$ qubits using the same entangled resource state and the same message size. 
    An $\epsilon$-correct and $\delta$-secure CDQS protocol hiding secret $Q$ using a $n_E$ qubit resource state $n_M$ qubit messages implies the existence of a $\max\{\epsilon,2\sqrt{\delta} \}$-correct $f$-routing protocol that routes system $Q$ using $n_E$ qubits of resource state and $4(n_M+n_E)$ qubits of message. 
\end{theorem}
\begin{proof} \,Begin by considering an $f$-routing protocol. 
Figure \ref{fig:NLQCandCDQS} establishes the subsystem labels we will use here. 
We will first show that an $f$-routing protocol is easily modified to construct a CDQS protocol. 
To do so, we send systems $M_0$ and $M_1$ that Bob would receive in the second round of the $f$-routing protocol to the referee of the CDQS protocol. 
Then, if $f(x,y)=1$, $\epsilon$-correctness of the $f$-routing scheme is immediately $\epsilon$-correctness of the CDQS.  

To show secrecy of the CDQS protocol, we first establish some notation. 
We label the channel realized by the first round operations of Alice and Bob $\mathbfcal{N}_{Q\rightarrow MM'}$, and let $\mathbf{V}_{Q\rightarrow MM'E}$ be a isometric extension of this channel. 
By correctness in $0$ instances of the $f$-routing scheme, we have that there exists a channel $\mathbfcal{D}^{xy}_{M'\rightarrow Q}$ such that
\begin{align}
    ||\mathbfcal{D}^{xy}_{M'\rightarrow Q}\circ [\tr_M \circ \mathbfcal{N}_{Q\rightarrow M'M}^{x,y}] - \mathbfcal{I}_Q ||_\diamond = ||\mathbfcal{D}^{xy}_{M'\rightarrow Q}\circ [\tr_{ME} (\mathbf{V}^{xy}_{Q\rightarrow MM'E} \cdot (\mathbf{V}^{xy}_{Q\rightarrow MM'E})^\dagger)] - \mathbfcal{I}_Q ||_\diamond \leq \epsilon \nonumber 
\end{align}
Then the decoupling theorem \ref{thm:newdecoupling} tells us that there exists a completely depolarizing channel $\mathbfcal{S}_{Q\rightarrow ME}$ such that
\begin{align}
    ||\tr_{M'} (\mathbf{V}^{xy}_{Q\rightarrow MM'E} \cdot \mathbf{V}^{xy}_{Q\rightarrow MM'E}) - \mathbfcal{S}^{xy}_{Q\rightarrow ME}||_\diamond \leq 2\sqrt{\epsilon}
\end{align}
Adding a trace over part of the outputs of channels can only make the channels less distinguishable, and hence the diamond norm smaller, so that
\begin{align}
    ||\tr_{M'E} (\mathbf{V}^{xy}_{Q\rightarrow MM'E} \cdot \mathbf{V}^{xy}_{Q\rightarrow MM'E}) - \mathbfcal{S}^{xy}_{Q\rightarrow M}||_\diamond \leq 2\sqrt{\epsilon}
\end{align}
but this is just
\begin{align}
    ||\mathbfcal{N}^{xy}_{Q\rightarrow M} - \mathbfcal{S}^{xy}_{Q\rightarrow M}||_\diamond \leq 2\sqrt{\epsilon}
\end{align}
which is exactly $2\sqrt{\epsilon}$-security of the CDQS. 
Note that the CDQS protocol defined by the $f$-routing protocol uses the same entangled resource state and no more communication. 

Now suppose we have a CDQS protocol which is $\epsilon$-correct and $\delta$-secure. 
Then to build the $f$-routing protocol, purify the channels Alice and Bob perform to isometries, and send the original message systems of the CDQS to Bob and their purifications to Alice. Then by $\epsilon$-correctness of the CDQS protocol, we immediately have $\epsilon$-correctness of the $f$-routing protocol when $f(x,y)=1$.

Next consider the case where $f(x,y)=0$. 
Then security of the CDQS implies that there exists a simulator channel $\mathbfcal{S}_{\varnothing \rightarrow M}^{xy}$ such that
\begin{align}
    ||\mathbfcal{S}_{\varnothing \rightarrow M}^{xy} \circ \tr_Q - \mathbfcal{N}^{xy}_{Q\rightarrow M} ||_\diamond \leq \delta
\end{align}
We will again apply the decoupling theorem. 
Notice that now, because of how we have defined the $f$-routing protocol, the map from $Q$ to $MM'$ is isometric, so $(\mathbfcal{N}^{xy})^c_{Q\rightarrow M} = (\mathbfcal{N}^{xy})_{Q\rightarrow M'}$.
Then the decoupling theorem implies the existence of a decoding channel $\mathbfcal{D}_{M'\rightarrow Q}^{xy}$ such that
\begin{align}
    ||\mathbfcal{D}^{xy}_{M'\rightarrow Q} \circ (\mathbfcal{N}^{xy})^c_{Q\rightarrow M'} -\mathbfcal{I}_Q||_\diamond \leq \sqrt{4 || \mathbfcal{S}_{\varnothing \rightarrow M}^{xy} \circ \tr_Q - \mathbfcal{N}^{xy}_{Q\rightarrow M}||} \leq 2\sqrt{\delta}
\end{align}
which gives $2\sqrt{\delta}$ correctness on $0$ instances. 
The protocol is then $\max\{2\sqrt{\delta},\epsilon \}$-correct. 

To see how the communication in the resulting $f$-routing protocol is related to the communication in the original CDQS protocol, we can use that a channel $\mathbfcal{N}_{A\rightarrow B}$ can always be purified by an isometry $\mathbf{V}_{A\rightarrow B C}$ where $d_C\leq d_Ad_B$. 
Let CDQS have messages that each consist of at most $n_{M}$ qubits, and use an $n_{E}$ qubit resource system on systems $LR$.
Then the most general possible protocol is defined by families of channels 
\begin{align}
\{\mathbfcal{N}^x_{L\rightarrow M_{0}}\},\,\,\,\, \{\mathbfcal{N}^y_{R\rightarrow M_{1}}\}
\end{align}
applied on the left and right respectively. 
We define purifications of these, 
\begin{align}
\{\mathbf{V}^x_{L\rightarrow M_{0}M_{0}'}\},\,\,\,\, \{\mathbf{V}^y_{R\rightarrow M_{1}M_1'}\}
\end{align}
We see that the message sizes are now at most $n_M + n_E$
qubits, so the total size of the communication is at most $4(n_M+n_E)$.
The entangled resource system used in the $f$-routing protocol is identical to the one used in the CDQS.
\end{proof}

\vspace{0.2cm}
\noindent \textbf{Explicit reconstruction procedure:}
\vspace{0.2cm}

It is perhaps counter-intuitive that the $f$-routing protocol built from the CDQS protocol succeeds in the case when $f(x,y)=0$. 
This is implied by the general physics of decoupling as captured by theorem \ref{thm:newdecoupling}, but for intuition we give a more explicit description in a special case here. 

Let's suppose the CDQS protocol is perfectly correct, and works in the following way. 
Assume the quantum secret is a single qubit and is stored in system $Q$.
To hide the quantum state on $Q$, Alice applies the one-time pad using a classical string $s=(s_1,s_2)$ as key. 
Explicitly she has applied
\begin{align}
    \ket{s_1,s_2}_A \ket{\psi}_Q\rightarrow \ket{s_1,s_2}_A (i)^{s_1\cdot s_2} X^{s_1}Z^{s_2}\ket{\psi}_Q.
\end{align}
A message system $M$ is sent to Bob, which reveals the key if and only if $f(x,y)=1$. The system $A$ must be sent to Alice on the left. 
The full state of the message systems then has the form
\begin{align}
    \frac{1}{2}\sum_{s_1,s_2,m_L,m_R} p(m_L,m_R|x,y,s) \ket{m_L}_{M'}\ket{s_1,s_2}_A (i)^{s_1\cdot s_2} X^{s_1}Z^{s_2}\ket{\psi}_Q \ket{m_R}_{M}.
\end{align}
Suppose we are in the case where $f(x,y)=0$. 
Then by security, the state on $M$ is independent of $s$. 
We can trace it out and the $M'$ system out and obtain the pure state
\begin{align}
    \frac{1}{2}\sum_{s_1,s_2} \ket{s_1,s_2}_A (i)^{s_1\cdot s_2} X^{s_1}Z^{s_2}\ket{\psi}_Q .
\end{align}
The claim is that Alice can recover the state on $Q$ from the $A$ system. 
To do this, she maps $\ket{s_1,s_2}$ to the Bell basis, obtaining
\begin{align}
    \frac{1}{2}(III+IXX+IZZ+IYY) \ket{\Psi^+}_{A_1A_2}\ket{\psi}_Q.
\end{align}
Then notice that
\begin{align}
    \frac{1}{2}(I_{A_2}I_{Q}+X_{A_2}X_{Q}+Z_{A_2}Z_{Q}+Y_{A_2}Y_{Q}) = SWAP_{A_2Q}
\end{align}
so that mapping $A_1A_2$ into the Bell basis actually swaps the state on $Q$ into $A_2$, so that Alice recovers the state on $Q$. 

\subsection{PSQM gives CDQS}

Analogous to the observation that PSM gives CDS, we can also show that PSQM gives CDQS. 
\begin{theorem}\label{thm:PSQMtoCDQS}
Suppose that an $\epsilon$-correct and $\delta$-private PSQM protocol exists for $f(x,y)\in\{0,1\}$ using messages of at most $n_M$ bits and an entangled state of no more than $n_E$ qubits. 
Then there exists a CDQS protocol hiding one qubit using $n_M+1$ bits of message and $n_E$ qubits of entangled state which is $2\epsilon$ correct and $\delta$ private. 
\end{theorem}
\begin{proof}\,
If the function $f(x,y)$ is constant then the CDQS protocol is trivial, so we assume without loss of generality that $f(x,y)$ is non-constant. 

Given the PSQM protocol, we build a CDQS protocol as follows. 
We introduce two random shared bits which we call $s=(s_1,s_2)$, which are held by Alice and Bob. 
Alice and Bob also pre-agree on a pair of inputs $(x,y)$ where $f(x,y)=0$, call them $(x_*,y_*)$, which exist because $f$ is non-constant by assumption. 
Upon receiving inputs $x,y$ Alice and Bob compute 
\begin{align}
    x_i' &= s_i x + (1-s_i) x_* \nonumber \\
    y_i' &= s_i y + (1-s_i) y_*
\end{align}
for $i=1,2$.
They run the PSQM protocol for $f$ on inputs $(x_1,y_1)$ and $(x_2,y_2)$ in parallel. 
Note that following the remark made after definition \ref{def:PSQM}, the PSQM for $F(x,y,s) = (f(x_1,y_1),f(x_2,y_2))$ is $2\epsilon$ correct and $2\delta$ secure. 
Notice that
\begin{align}\label{eq:fsrelation}
    f(x_i',y_i') = f(x,y)\wedge s_i
\end{align}
This means that by running the PSM for $f(x_i',y_i')$, the referee will learn $s_i$ when $f(x,y)=1$.
In the CDQS protocol, we have Alice act on the quantum secret $Q$ with the one time pad using the key $s=(s_1,s_2)$. 
Then the referee will be able to undo the one time pad when $f(x,y)=1$ (and so they know $s$), but not otherwise. 

Next we establish correctness more carefully. 
First note that the encoding channel for the CDQS defined by the above protocol is
\begin{align}
    \mathbfcal{N}_{Q\rightarrow MQ}^{xy}(\cdot) = \frac{1}{2^{|s|}}\sum_s P^s_Q \cdot P^s_Q \otimes \rho_M(x,y,s),
\end{align}
where $\rho_M$ is the state of the message systems prepared by the PSQM. 
Correctness of the CDQS requires we establish the existence of a channel which approximately inverts this. 
Note that by $2\epsilon$ correctness of the PSQM, we have that there exists a channel $\mathbfcal{V}_{M\rightarrow Z}$ such that
\begin{align}\label{eq:PSQMsecuritysuccint}
    ||\mathbfcal{V}_{M\rightarrow Z}(\rho_M(x,y,s)) -\ketbra{F'}{F'}_Z||_1\leq 2\epsilon
\end{align} 
where we defined $F'=(f(x_1',y_1'),f(x_2',y_2'))$.
We define our decoding channel to apply $\mathbfcal{V}_{M\rightarrow Z}$, measure the $Z$ system, then apply a Pauli conditioned on the outcome, 
\begin{align}
    \mathbfcal{D}_{MQ\rightarrow Q}(\cdot) = \sum_{F} P_Q^{F}\otimes \bra{F}_Z\mathbfcal{V}_{M\rightarrow Z}(\cdot) \ket{F}_Z\otimes P_Q^{F}.
\end{align}
We claim this is an approximate inverse to $\mathbfcal{N}_{Q\rightarrow MQ}^{xy}$. 
Using the definitions of $\mathbfcal{N}_{Q\rightarrow MQ}^{xy}$,  $\mathbfcal{D}_{MQ\rightarrow Q}$ and the diamond norm, we obtain
\begin{align}
    ||\mathbfcal{D}_{MQ\rightarrow Q}&\circ \mathbfcal{N}_{Q\rightarrow MQ}^{xy} - \mathbfcal{I}_Q ||_\diamond \nonumber \\
    &= \sup_n \max_{\Psi_{R_nQ}} || \frac{1}{2^{|s|}} \sum_{s,F} P^{s+F}_Q \Psi_{R_nQ}P^{s+F}_Q \otimes \bra{F}_Z\mathbfcal{V}_{M\rightarrow\bar{M}Z}(\rho_{M}(x,y,s))\ket{F}_Z - \Psi_{R_nQ}||_1 \nonumber \\
    &\leq 2\epsilon + \sup_n \max_{\Psi_{R_nQ}} || \frac{1}{2^{|s|}} \sum_{s,F} P^{s+F}_Q \Psi_{R_nQ}P^{s+F}_Q \otimes \bra{F}_Z\ketbra{F'}{F'}\ket{F}_Z - \Psi_{R_nQ}||_1 \nonumber 
\end{align}
where we replaced the $\mathbfcal{V}_{M\rightarrow\bar{M}Z}(\rho_{M}(x,y))$ with $\ketbra{F'}{F'}$ at the expense of the added $2\epsilon$, which is justified by equation \ref{eq:PSQMsecuritysuccint}. 
Continuing, we can see that the second term is actually zero, since (from equation \ref{eq:fsrelation}) $F'$ is just $s$ when $f(x_1,y_1)=f(x_2,y_2)=1$, which removes the Pauli's and so the full diamond norm is bounded by $2\epsilon$. 

Next we study security of the CDQS protocol. 
Recall that security of the PSQM implies that there exists a channel $\mathbfcal{S}_{Z\rightarrow M}$ such that
\begin{align}\label{eq:rhoV}
    ||\rho_M(x,y,s) - \mathbfcal{S}_{Z\rightarrow M}(\ketbra{F'}{F'})||_1 \leq \delta.
\end{align}
In the definition of security for CDQS, we need to show the existence of a channel $\mathbfcal{S'}^{x,y}_{\varnothing\rightarrow M}$ such that $\mathbfcal{S'}^{x,y}_{\varnothing\rightarrow M} \circ \tr_Q$ is close to the action of the protocol $\mathbfcal{N}_{Q\rightarrow MQ}^{xy}$. 
We define
\begin{align}
    \mathbfcal{S'}^{x,y}_{\varnothing\rightarrow MQ} = \mathbfcal{S}_{Z\rightarrow M}(\ketbra{F'}{F'}) \otimes \frac{\mathcal{I}_Q}{d_Q},
\end{align}
then consider, 
\begin{align}
    ||\mathbfcal{S'}^{x,y}_{\varnothing\rightarrow MQ}\circ \tr_Q &-\mathbfcal{N}_{Q\rightarrow MQ}^{xy}||_\diamond \nonumber\\
    &= \sup_n \max_{\Psi_{R_nQ}} || \frac{\mathcal{I}_Q}{d_Q}\otimes \mathbfcal{S}_{Z\rightarrow M}(\ketbra{F'}{F'}) - \frac{1}{2^{|s|}} \sum_s P^s_Q \Psi_{R_nQ}P^s_Q \otimes \rho_M(x,y,s)||_1 \nonumber \\
    &\leq \sup_n \max_{\Psi_{R_nQ}} || \frac{\mathcal{I}_Q}{d_Q}\otimes \mathbfcal{S}_{Z\rightarrow M}(\ketbra{F'}{F'}) - \frac{1}{2^{|s|}} \sum_s P^s_Q \Psi_{R_nQ}P^s_Q \otimes \mathbfcal{S}_{Z\rightarrow M}(\ketbra{F'}{F'})||_1 + \delta \nonumber \\
    &= \delta \nonumber 
\end{align}
where we used \ref{eq:rhoV} in the inequality. This is $\delta$ security of the CDQS.

\end{proof}

\subsection{CFE gives PSQM and weak converse}

Finally, we relate coherent function evaluation to PSQM. 
Note that the relationship is only that good CFE protocols give good PSQM protocols, although a weak converse also exists, as we describe. 

\begin{figure*}
    \centering
    \begin{subfigure}{0.45\textwidth}
    \centering
    \begin{tikzpicture}[scale=0.4]
    
    \draw[thick] (-5,-5) -- (-5,-3) -- (-3,-3) -- (-3,-5) -- (-5,-5);
    
    \draw[thick] (5,-5) -- (5,-3) -- (3,-3) -- (3,-5) -- (5,-5);
    
    \draw[thick] (5,5) -- (5,3) -- (3,3) -- (3,5) -- (5,5);
    
    \draw[thick] (4.5,-3) -- (4.5,3);
    
    \draw[thick] (-3.5,-3) to [out=90,in=-90] (3.5,3);
    
    \draw[thick] (-3.5,-5) to [out=-90,in=-90] (3.5,-5);
    \draw[black] plot [mark=*, mark size=3] coordinates{(0,-7.05)};
    
    \draw[thick] (-4.5,-6) -- (-4.5,-5);
    \node[below] at (-4.5,-6) {$x$};
    
    \draw[thick] (4.5,-6) -- (4.5,-5);
    \node[below] at (4.5,-6) {$y$};

    \node[left] at (-1,-2) {$M_0$};
    \node[right] at (4.5,0) {$M_1$};
    

    \draw[thick] (4.5,5) -- (4.5,6);
    \node[above] at (4.75,6) {$\tilde{M}$};

    \draw[thick] (3.5,5) -- (3.5,6);
    \node[above] at (3.25,6) {$Z$};

    \node at (-4,-4) {$\mathbf{\mathbfcal{N}}^L$};
    \node at (4,-4) {$\mathbf{\mathbfcal{N}}^R$};
    \node at (4,4) {$\mathbf{W}^R$};

    \node at (-2,-6) {$C'$};
    \node at (2,-6) {$C$};
    
    \end{tikzpicture}
    \caption{}
    \label{fig:PSQM}
    \end{subfigure}
    \hfill
    \begin{subfigure}{0.45\textwidth}
    \centering
    \begin{tikzpicture}[scale=0.4]
    
    \draw[thick] (-5,-5) -- (-5,-3) -- (-3,-3) -- (-3,-5) -- (-5,-5);
    \node at (-4,-4) {$\mathbf{V}^L$};
    
    \draw[thick] (5,-5) -- (5,-3) -- (3,-3) -- (3,-5) -- (5,-5);
    \node at (4,-4) {$\mathbf{V}^R$};
    
    \draw[thick] (5,5) -- (5,3) -- (3,3) -- (3,5) -- (5,5);
    \node at (4,4) {$\mathbf{W}^R$};
    
    \draw[thick] (-5,5) -- (-5,3) -- (-3,3) -- (-3,5) -- (-5,5);
    \node at (-4,4) {$\mathbf{W}^L$};
    
    \draw[thick] (-4.5,-3) -- (-4.5,3);
    
    \draw[thick] (4.5,-3) -- (4.5,3);
    
    \draw[thick] (-3.5,-3) to [out=90,in=-90] (3.5,3);
    
    \draw[thick] (3.5,-3) to [out=90,in=-90] (-3.5,3);
    
    \draw[thick] (-3.5,-5) to [out=-90,in=-90] (3.5,-5);
    \draw[black] plot [mark=*, mark size=3] coordinates{(0,-7.05)};
    
    \draw[thick] (-4.5,-6) -- (-4.5,-5);
    \draw[thick] (4.5,-6) -- (4.5,-5);
    
    \draw[thick] (4.5,5) -- (4.5,6);
    \node[above] at (4.75,6) {$\tilde{M}$};
    
    \draw[thick] (-4.5,5) -- (-4.5,6);
    \node[above] at (-4.75,6) {$\tilde{M}'$};
    
    \draw[thick] (3.5,5) -- (3.5,6);
    \node[above] at (3.25,6) {$Z$};
    
    \draw[thick] (-3.5,5) -- (-3.5,6);
    \node[above] at (-3.25,6) {$Z'$};

    \node[left] at (-1,-2) {$M_0$};
    \node[right] at (4.5,0) {$M_1$};

    \node[right] at (1,-2) {$M_1'$};
    \node[left] at (-4.5,0) {$M_0'$};

    \node at (-2,-6) {$C'$};
    \node at (2,-6) {$C$};

    \node[below] at (-4.5,-6) {$X$};
    \node[below] at (4.5,-6) {$Y$};
 
    \end{tikzpicture}
    \caption{}
    \label{fig:CFE}
    \end{subfigure}
    \caption{Corresponding PSQM (left) and CFE (right) protocols, with labellings of the subsystems involved shown.}
    \label{fig:PSQMandCFE}
\end{figure*}
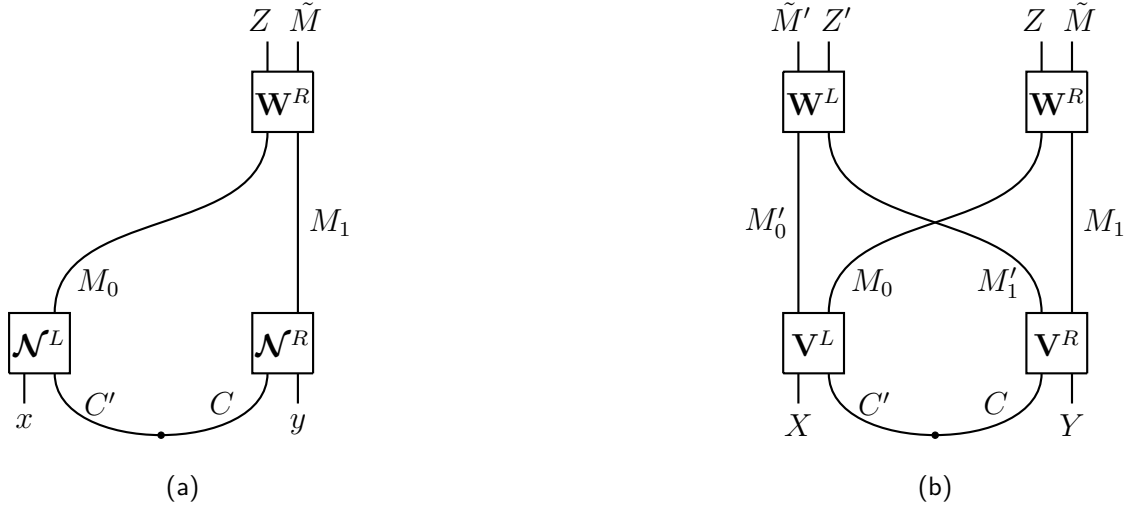

\begin{theorem}\label{thm:CFEtoPSQM}
    An $\epsilon$-correct CFE protocol for the function $f$ using $n_E$ EPR pairs and messages of $n_M$ qubits implies the existence of an $\epsilon$-correct and $\sqrt{\epsilon}$-secure PSQM protocol for the same function, using $n_E$ EPR pairs and no more than $n_M$ message qubits.  
\end{theorem}
\begin{proof}
\,We define the PSQM protocol from the CFE protocol as follows. 
The PSQM protocol uses the same resource state as the CFE, Alice applies the bottom left operation of the CFE, Bob applies the bottom right operation of the CFE, and they send the systems that would reach the top right of the CFE protocol to the referee, which we call the $M$ systems. 
To produce their output, the referee applies the top right operation from the CFE. 
See figure \ref{fig:PSQMandCFE} for labels of the relevant subsystems. 

Correctness of the CFE protocol means that we have
\begin{align}
    ||\mathbf{F}(\cdot)\mathbf{F}^\dagger - \mathbf{\mathbfcal{N}}_{XY\rightarrow Z'Z}||_\diamond \leq \epsilon
\end{align}
where $\mathbf{\mathbfcal{N}}$ is the channel applied by our CFE protocol and $\mathbf{F}$ denotes the CFE isometry to be implemented. 
Applying these channels to the input $\ket{x}_X\ket{y}_Y$ and using the definition of the diamond norm distance, we obtain
\begin{align}
    ||\ketbra{xy}{xy}_{Z'}\otimes \ketbra{f_{xy}}{f_{xy}}_{Z} - \rho_{Z'Z}(x,y)||_1 \leq \epsilon.
\end{align}
Tracing out the $Z'$ system and using that the one norm distance decreases under the partial trace, we obtain $\epsilon$-correctness of the PSQM. 

Next we study security of the PSQM. 
We start again from the correctness of the CFE protocol. 
To simplify our notation, we define the channels (see also figure \ref{fig:PSQMandCFE})
\begin{align}
    \mathbfcal{F}_{XY\rightarrow Z'Z}(\cdot) &= \mathbf{F} (\cdot) \mathbf{F}^\dagger, \nonumber \\
    \mathbfcal{W}^L_{M\rightarrow Z\tilde{M} }(\cdot) &= \mathbf{W}^R_{M\rightarrow Z\tilde{M}} (\cdot) (\mathbf{W}^R_{M\rightarrow Z\tilde{M}})^\dagger \nonumber\\
    \mathbfcal{W}^R_{M'\rightarrow Z'\tilde{M}' }(\cdot) &= \mathbf{W}^L_{M'\rightarrow Z'\tilde{M}'} (\cdot) (\mathbf{W}^L_{M'\rightarrow Z'\tilde{M}'})^\dagger \nonumber\\
    \mathbfcal{W}_{MM'\rightarrow Z\tilde{M} Z'\tilde{M}'} &= \mathbfcal{W}^L_{M'\rightarrow Z\tilde{M}'}\otimes \mathbfcal{W}^R_{M\rightarrow Z\tilde{M}} \nonumber\\
    \mathbfcal{V}_{XY\rightarrow MM'}(\cdot) &=  \mathbf{V}^R_{YC\rightarrow M_1'M_1} \otimes \mathbf{V}^L_{XC'\rightarrow M_0M_0'}(\cdot \otimes \Psi_{CC'}) (\mathbf{V}^R_{YC\rightarrow M_1'M_1} \otimes \mathbf{V}^L_{XC'\rightarrow M_0M_0'})^\dagger \nonumber
\end{align}
Then we note that the CFE protocol can be decomposed into two steps, and rewrite the statement of correctness,
\begin{align}
    ||\mathbfcal{F}_{XY\rightarrow Z'Z}(\cdot) - \tr_{\tilde{M}\tilde{M}'} (\mathbfcal{W}_{MM'\rightarrow Z\tilde{M}Z'\tilde{M}'}) \circ (\mathbfcal{V}_{XY\rightarrow MM'})||_\diamond \leq \epsilon.\nonumber 
\end{align}
Next, we will use that Stinespring dilations of channels can be chosen to be close if the initial channels are close \cite{kretschmann2008continuity}. 
In particular we have
\begin{align}\label{eq:opunderdiamond}
    \frac{||T_1-T_2||_\diamond}{\sqrt{||T_1||_\diamond}+\sqrt{||T_2||_\diamond}} \leq \inf_{V_1,V_2}||V_1-V_2 ||_{op} \leq \sqrt{||T_1-T_2||_{\diamond}}
\end{align}
where the infimum is over all dilations $V_i$ of $T_i$. 
Noting that $\mathbfcal{F}$ is already isometric, we have that its dilations must consist of adding a state preparation channel, which we label $\mathbfcal{P}_{\emptyset\rightarrow E}$. 
Further, all dilations are related by a partial isometry on the auxiliary space, so the dilations of the $\tr_{\tilde{M}\tilde{M}'}\mathbfcal{W}\circ \mathbfcal{V}$ channel can be written in the form
\begin{align}
    \mathbfcal{U}_{XY\rightarrow ZZ'E} = \mathbfcal{I}_{\tilde{M}\tilde{M}'\rightarrow E} \circ (\mathbfcal{W}_{MM'\rightarrow Z\tilde{M}Z'\tilde{M}'}) \circ (\mathbfcal{V}_{XY\rightarrow MM'})
\end{align}
Then using the upper bound in \ref{eq:opunderdiamond}, we have
\begin{align}\label{eq:operatornormundersqrtepsilon}
    ||\mathbfcal{F}_{XY\rightarrow Z'Z}\otimes \mathbfcal{P}_{\emptyset \rightarrow E} - \mathbfcal{I}_{\tilde{M}\tilde{M}'\rightarrow E}\circ \mathbfcal{W}_{MM'\rightarrow Z\tilde{M}Z'\tilde{M}'} \circ \mathbfcal{V}_{XY\rightarrow MM'}||_{op} \leq \sqrt{\epsilon}.
\end{align}
Next, we will exploit the lower bound in \ref{eq:opunderdiamond} to translate this to an upper bound on the diamond norm of these isometries. 
To do this, notice that from \ref{eq:opunderdiamond} we have
\begin{align}
    \frac{||V_1-V_2||_\diamond}{\sqrt{||V_1||_\diamond}+\sqrt{||V_2||_\diamond}} \leq \inf_{P_1,P_2}||V_1\otimes P_1-V_2\otimes P_2 ||_{op} \leq ||V_1-V_2 ||_{op}
\end{align}
Using this in equation \ref{eq:operatornormundersqrtepsilon}, we obtain
\begin{align}
    ||\mathbfcal{F}_{XY\rightarrow Z'Z}\otimes \mathbfcal{P}_{\emptyset \rightarrow E} - \mathbfcal{I}_{\tilde{M}\tilde{M}'\rightarrow E}\circ \mathbfcal{W}_{MM'\rightarrow Z\tilde{M}Z'\tilde{M}'} \circ \mathbfcal{V}_{XY\rightarrow MM'}||_{\diamond} \leq 2\sqrt{\epsilon}.
\end{align}
Next, apply $\mathbfcal{I}^\dagger_{\tilde{M}\tilde{M}'\rightarrow E}$ to both terms, which cannot increase the diamond norm, and obtain
\begin{align}
    ||\mathbfcal{F}_{XY\rightarrow Z'Z}\otimes \mathbfcal{P}_{\emptyset \rightarrow \tilde{M}\tilde{M}'} -  \mathbfcal{W}_{MM'\rightarrow Z\tilde{M}Z'\tilde{M}'} \circ \mathbfcal{V}_{XY\rightarrow MM'}||_{\diamond} \leq 2\sqrt{\epsilon}.
\end{align}
Apply $\mathbfcal{W}^\dagger_{MM'\rightarrow Z\tilde{M}Z'\tilde{M}'}$ to both terms to obtain
\begin{align}
    ||\mathbfcal{W}_{MM'\rightarrow Z\tilde{M}Z'\tilde{M}'}^\dagger \circ (\mathbfcal{F}_{XY\rightarrow Z'Z}\otimes \mathbfcal{P}_{\emptyset \rightarrow \tilde{M}\tilde{M}'}) - \mathbfcal{V}_{XY\rightarrow MM'}||_{\diamond} \leq 2\sqrt{\epsilon}.
\end{align}
Then, apply these channels to the input $\ket{xy}_{XY}$ and call the output of the protocol on the $M$ system $\rho_{M}(x,y)$, and trace out the $\tilde{M}'$ system,
\begin{align}
    ||\tr_{M'}\mathbfcal{W}_{MM'\rightarrow Z\tilde{M}Z'\tilde{M}'}^\dagger \circ \mathbfcal{F}_{XY\rightarrow Z'Z}(\ketbra{xy}{xy})\otimes \psi_{\tilde{M}\tilde{M}'} - \rho_{M}(x,y)||_1 \leq 2\sqrt{\epsilon}. \nonumber
\end{align}
Simplifying the state on the left using
\begin{align}
    \mathbfcal{W}_{MM'\rightarrow Z\tilde{M}Z'\tilde{M}'} &= \mathbfcal{W}^L_{M\rightarrow Z\tilde{M}}\otimes \mathbfcal{W}^R_{M'\rightarrow Z'\tilde{M}'} \nonumber \\
    \mathbfcal{F}_{XY\rightarrow Z'Z}(\ketbra{xy}{xy}) &= \ketbra{f_{xy}}{f_{xy}}_Z\otimes \ketbra{xy}{xy}_{Z'}
\end{align}
we obtain
\begin{align}
    ||\mathbfcal{W}^{R\dagger}_{M\rightarrow \tilde{M}Z}(\ketbra{f_{xy}}{f_{xy}}\otimes \sigma_{\tilde{M}} ) - \rho_{M}(x,y)||_1 \leq 2\sqrt{\epsilon}
\end{align}
which is $2\sqrt{\epsilon}$ security of the PSQM protocol, where $\mathbfcal{W}^{R\dagger}_{M\rightarrow \tilde{M}Z}$ along with the state preparation of $\sigma_{\tilde{M}}$ defines the simulator channel. 
\end{proof}

Next, we give a weak converse to the above theorem, which shows that a good PSQM protocol implies the existence of CFE protocol that succeeds with constant probability when acted on the maximally entangled state.
Note that this falls short of bounding the diamond norm. 
We show this only in the exact setting though a robust version might also exist.
We are also limited to the case where the function outputs a single bit.
\begin{theorem} \label{thm:PSQMtoCFEWeak}
    Suppose there exists a perfectly correct and perfectly secure PSQM protocol for the function $f:X\times Y\rightarrow Z$ with $Z\in\{0,1\}$ using $n_M$ bits of communication and $n_E$ qubits of entangled resource system.
    Then there is a CFE protocol that implements a channel $\tilde{\mathbfcal{V}}^f_{XY\rightarrow Z'Z}$ such that
    \begin{align}
        F(\tilde{\mathbfcal{V}}^f_{XY\rightarrow Z'Z}(\Psi^+_{RXY}),\mathbf{V}^f_{XY\rightarrow Z'Z}(\Psi^+)_{RXY}(\mathbf{V}^f_{XY\rightarrow Z'Z})^\dagger) \geq \frac{1}{2}
    \end{align}
    and which uses $n_E$ qubits of entangled resource state and $n_M+n_E+2n$ qubits of communication, where $n$ is the input size. 
\end{theorem}
\begin{proof}\,
    By security of the PSQM protocol, we have that when given input $\ket{xy}$ the protocol produces a reduced state $\rho_{M}(x,y)$ with the form
    \begin{align}
        \rho_{M}(x,y) = \mathbfcal{S}_{Z\rightarrow M}(\ketbra{f_{xy}}{f_{xy}}) = \sigma^{f_{xy}}_{M}.
    \end{align}
    As part of the CFE protocol we are defining, we make a copy of the inputs $\ket{x}_X\ket{y}_Y$ and send this copy in a system labelled $Z'$ to the left. 
    The overall state of the message system then is,
    \begin{align}
        \ketbra{xy}{xy}_{Z'}\otimes \sigma^{f_{xy}}_{M}.
    \end{align}
    Now consider purifying the channels used in the PSQM protocol, and sending the purifying systems (call them $\tilde{M}'$) to the left. 
    Then the message system becomes
    \begin{align}
        \ket{\Psi_{xy}}_{Z'\tilde{M}'M} = \ket{xy}_{Z'} \sum_k \sqrt{\lambda^k_{f_{xy}}}\ket{\psi^k_{f_{xy}}}_{\tilde{M}'}\ket{\psi^k_{f_{xy}}}_{M}
    \end{align}
    where we used that the reduced density matrix on $M$ depends only on $f_{xy}$ to enforce that the Schmidt coefficients and Schmidt vectors on $M$ can depend only on $f_{xy}$. 

    Next, we consider adding to the protocol a unitary
    \begin{align}
        \mathbf{U}_{Z'\tilde{M}'} = \sum_{x,y,k} \alpha_{f_{xy}} \ketbra{xy}{xy}_{Z'}\otimes \ketbra{k}{\psi^k_{f_{xy}}}_{\tilde{M}'}
    \end{align}
    where the $\alpha_{f_{xy}}$ are phases, $|\alpha_{f_{xy}}|^2=1$. 
    We will determine later how to choose these phases.
    This means we produce the state
    \begin{align}\label{eq:PSQMstatesofar}
        \mathbf{U}_{Z'\tilde{M}'} \ket{\Psi_{xy}}_{Z'\tilde{M}'M} = \ket{xy}_{Z'} \sum_k \alpha_{f_{xy}}\sqrt{\lambda^k_{f_{xy}}} \ket{k}_{\tilde{M}'}\ket{\psi^k_{f_{xy}}}_{M}.
    \end{align}
    We'd like to exploit the correctness of the PSQM protocol to show this state can be made, using an operation on $M$, to have large overlap with the correct output for the CFE protocol, which here is $\ket{xy}_{Z'}\ket{f_{xy}}_{Z}$.
    Looking at the reduced state on $M$ again, we have 
    \begin{align}
        \sigma_{M} = \sum_k \lambda^k_{f_{xy}}\ketbra{\psi_{f_{xy}}^k}{\psi_{f_{xy}}^k}_{M} .
    \end{align}
    From correctness we have that there exists a map $\mathbf{V}_{M\rightarrow \tilde{M}Z}$ such that
    \begin{align}
        \sum_k \lambda^k_{f_{xy}} \tr_{\tilde{M}}(\mathbf{V}_{M\rightarrow \tilde{M}Z}\ketbra{\psi_{f_{xy}}^k}{\psi_{f_{xy}}^k}_{M}\mathbf{V}^\dagger_{M\rightarrow \tilde{M}Z}) = \ketbra{f_{xy}}{f_{xy}}_{Z}
    \end{align}
    which is only solved if, for all $k$,  
    \begin{align}
        \mathbf{V}_{M\rightarrow \tilde{M}Z}\ket{\psi^k_{f_{xy}}}_{M} = \beta^{k}_{f_{xy}} \ket{f_{xy}}_{Z} \ket{\tilde{\psi}^k_{f_{xy}}}_{\tilde{M}}
    \end{align}
    with $\beta^k_{f_{xy}}$ being pure phases, $|\beta^k_{f_{xy}}|^2=1$.
    Returning to the form \ref{eq:PSQMstatesofar}, we can now add an application of $\mathbf{V}_{M\rightarrow Z\tilde{M}}$ as the top right element of our CFE protocol and we see that we produce the state
    \begin{align}
        \mathbf{V}_{M\rightarrow \tilde{M}Z} \mathbf{U}_{Z'\tilde{M}'}\ket{\Psi_{xy}}_{Z'\tilde{M}'M} &= \alpha_{f_{xy}} \ket{xy}_{Z'} \ket{f_{xy}}_{Z} \sum_k \beta^{k}_{f_{xy}}\sqrt{\lambda^k_{f_{xy}}}\ket{k}_{\tilde{M}'}\ket{\tilde{\psi}^k_{f_{xy}}}_{\tilde{M}} \nonumber \\
        &= \alpha_{f_{xy}} \ket{xy}_{Z'} \ket{f_{xy}}_{Z} \ket{\Phi_{f_{xy}}}_{\tilde{M}'\tilde{M}}.
    \end{align}
    By linearity, if we perform the same protocol on the state $\ket{\Psi^+}_{RXY}$ we produce the output
    \begin{align}
        \ket{\Psi'_f}_{RZ'Z\tilde{M}'\tilde{M}} = \frac{1}{\sqrt{d_R}}\sum_{xy} \alpha_{f_{xy}} \ket{xy}_R \ket{xy}_{Z'} \ket{f_{xy}}_Z \ket{\Phi_{f_{xy}}}_{\tilde{M}'\tilde{M}}.
    \end{align}
    We would like to compute the fidelity of the state produced by our protocol on $RZ'Z$ with the correct one when acted on the maximally entangled state. 
    Note that the correct output state would be
    \begin{align}
        \ket{\Psi_f} = \frac{1}{\sqrt{d_R}} \sum_{xy} \ket{xy}_R\ket{xy}_{Z'} \ket{f_{xy}}_{Z}.
    \end{align}
    Computing the fidelity of this with the partial state of $\ket{\Psi_f'}$ on $RZ'Z$, we find
    \begin{align}
        F(\Psi_f,\sigma)=\bra{\Psi_f}\sigma_{RZ'Z}\ket{\Psi_f} &= \frac{1}{d_R^2} \sum_{xy,x'y'} \alpha_{f_{xy}}^* \alpha_{f_{x'y'}} \braket{\Phi_{f_{xy}}}{\Phi_{f_{x'y'}}} 
    \end{align}
    Now, we can see how we should choose the phases $\alpha_{f_{xy}}$ that enter through our choice of the unitary $\mathbf{U}$.
    We should choose the phases such that this sum is lower bounded, which we can achieve by setting 
    \begin{align}
        \alpha_0 &= 1, \nonumber \\
        \alpha_1 &= \frac{\braket{\Phi_{1}}{\Phi_{0}}}{|\braket{\Phi_{1}}{\Phi_{0}}|}.
    \end{align}
    This ensures that the terms in the sum where $f_{xy}\neq f_{x'y'}$ are positive, so we bound them below by zero and obtain
\begin{align}
        F(\Psi_f,\sigma) &\geq \frac{1}{d_R^2} \left(\sum_{f_{xy}} \sum_{\substack{xyx'y':\\f_{xy}=f_{x'y'}}} 1 \right)\nonumber \\
        &= \frac{1}{d_R^2}\sum_{f_{xy}} N_{f_{xy}}^2 \geq \frac{1}{2}
    \end{align}
    where $N_{m}$ is the number of inputs that lead to $f_{xy}=m$. 
    This gives the needed lower bound. 
    
    To understand the resource consumption of the protocol constructed above, notice that it uses the same resource state, and so still $n_E$ qubits of entangled resource system. 
    Considering the message sizes, notice that in purifying the channels used in the PSQM protocol we need no more than $n_E+n_M$ qubits in the auxiliary system, and then we added an additional copy of the input sent to the left, so we use at most $n_E+2n+n_M$ qubit messages.  
\end{proof}

Recently a robust version of this theorem was proven, see \cite{koutsikos2024essay}. 

Note that we do not expect that a good PSQM protocol implies a CFE protocol that succeeds with fidelity near 1, even in the perfect case, and the above is likely the best implication from PSQM to CFE that is possible. 
To understand why, consider why the fidelity of $1/2$ appears in the above.
The security requirement of the PSQM implies that the density matrix on Bob's side in the CFE depends only on $f(x,y)$, and not further on $(x,y)$. 
In the proof above, this restricts the entanglement between $ZZ'$ and $\tilde{M}$, which can be exploited to make the CFE protocol mostly coherent. 
However, since the system on the right can still depend on $f(x,y)$, there can be one qubit worth of entanglement between $ZZ'$ and $\tilde{M}$, which leads to the fidelity of $1/2$. 
We do not believe there is any way to remove this last qubit of entanglement, since it seems consistent with the security of the PSQM, and hence no way to achieve fidelity 1 in the CFE. 
For this reason we should understand CFE as likely a stronger primitive than PSQM.
It would be interesting to understand if there is some other special case of NLQC, aside from CFE, which is equivalent to PSQM.

\section{Complexity of efficiently achievable functions}\label{sec:complexity}

The set of implications summarized in figure \ref{fig:web} imply efficient protocols for one primitive imply efficient protocols for many others. 
In this section we briefly summarize what is known about the efficiently achievable functions in various settings, and how they compare across various primitives. 

\subsection{Relevant complexity measures}\label{sec:complexitymeasures}

An important model of computation we will discuss is the modulo-$p$ branching program. 
These are computational models with close relationships to various non-uniform complexity classes sitting inside of NC. 
\begin{definition}
A \textbf{branching program} is a tuple $\mathbfcal{BP}=(G,\phi,s,t_0,t_1)$ where,
\begin{itemize}
    \item $G=(V,E)$ is a directed acyclic graph,
    \item $\phi$ is a function from edges in $E$ to either a value ``yes'' or a tuple $(b,i)$ for $b$ a bit and $i\in \{1,...,n\}$,
    \item $s$, $t_0$, $t_1$ are vertices from $V$. 
\end{itemize}
Given a $n$ bit string $x$ as input, the branching program specifies a subgraph of $G$ labelled $G_x$ according to the following rule. 
If for $e\in E$ we have $\phi(e)=(b,j)$ with $x_j=b$, or if $\phi(e)=$``yes'', then $e$ is included in $G_x$.
We define a function acc$(x)$ as the number of paths $s\rightarrow t_1$ in the graph $G_x$, and a function rej$(x)$ as the number of paths from $s$ to $t_0$ in $G_x$. 
\end{definition}

\begin{definition}
The size of a branching program is defined as the number of vertices in $V$. 
We label the minimal sized branching program computing $f$ as $BP(f)$. 
\end{definition}
We say a branching program is deterministic if the out degree of every vertex in every $G_x$ is at most $1$, and non-deterministic otherwise. 
The function $f(x)$ computed by a deterministic or non-deterministic branching program is defined such that $f(x)=1$ iff $acc(x)>0$. 
A \textbf{Boolean modulo-$p$ branching program} computes the function $f(x)$ defined such that $f(x)=1$ iff $acc(x)\neq 0 \,\,mod \,\, p$. 
We label the minimal size of a mod $p$ branching program computing $f$ by $BP_p(f)$. 

The class of functions with polynomial sized modulo-$p$ branching programs is defined below. 
\begin{definition}
    The complexity class Mod$_pL$/poly is defined as those Boolean function families $\{f_n\}$ which have polynomial (in $n$) sized modulo-$p$ branching programs. 
\end{definition}
The uniform complexity class Mod$_pL$ can be defined similarly in terms of log-space uniform branching programs, or given an equivalent definition in terms of Turing machines \cite{buntrock1992structure}. 
Another relevant complexity class, also based on branching programs, is the following. 
\begin{definition}
    The class $C_=L/poly$ (read as ``equality L'') is defined as those Boolean function families $\{f_n\}$ which can be decided in the following way. 
    We consider a branching program of polynomial (in n) size. 
    If $acc(x)=rej(x)$, output $1$ and otherwise output $0$.  
\end{definition}

A related notion of complexity that we will need is that of a \textbf{span program}, defined initially in \cite{karchmer1993span}.
\begin{definition}\label{def:spanprogram}
A \textbf{span program} over a field $\mathbb{Z}_p$ consists of a triple $S=(M, \phi, \mathbf{t})$, where $M$ is a $d\cross e$ matrix with entries in $\mathbb{Z}_p$, $\phi$ is a map from rows of $M$, labelled $r_i$, to pairs $(k,\varepsilon_i)$, with $k\in \{1,...,n\}$ and $\varepsilon_i\in\{0,1\}$, and $\mathbf{t}$ is a non-zero vector of length $e$ with entries in $\mathbb{Z}_p$.
A span program $S$ computes a function $f:\{0,1\}^n\rightarrow \{0,1\}$ as follows. 
Given an input string $z$ of $n$ bits, if the vector $\mathbf{t}$ is in $\text{span}(\{r_i: \exists j, \phi(r_i)=(j,z_j)\})$, then output 1. 
Otherwise, output 0.
\end{definition}

\begin{definition}
The \textbf{size} of a span program is defined to be $d$, the number of rows in $M$. 
We denote the minimal size of a span program over $\mathbb{Z}_p$ that computes $f$ by $SP_p(f)$.
\end{definition}
The size of a span program computing $\{f_n\}$ and of a branching program computing the same function family are related by the following theorem, noted in \cite{karchmer1993span} to follow from techniques in \cite{buntrock1992structure}.
\begin{theorem}
    For every prime $p$, Mod$_pL$ consists of those function families with polynomial sized span programs over $\mathbb{Z}_p$.
\end{theorem}
Thus the size of span programs and of arithmetic branching programs are related polynomially, and in fact \cite{beimel1999arithmetic}\footnote{Note that this statement is given in \cite{beimel1999arithmetic} in terms of \emph{arithmetic branching programs}, which are a generalization of modulo-p branching programs (and so are at least as powerful).}
\begin{align}\label{eq:SPvsBP}
    SP_p(f) \leq 2 BP_p(f).
\end{align}
We will never be interested in constant factor differences, so we can take that span programs are always smaller than modulo-$p$ branching programs. 

An important notion for us will be that of \emph{pre-processing}. 
We will consider functions $f:\{0,1\}^{n}\times \{0,1\}^{n}\rightarrow \{0,1\}$, and are interested in the complexity of computing $f(x,y)$ after allowing for arbitrary functions to be applied to $x$ and $y$ separately. 
We make the following definition. 
\begin{definition}
    An \textbf{interaction part} of $f(x,y):\{0,1\}^{n}\times \{0,1\}^{n}\rightarrow \{0,1\}$ is any function $F$ such that there exists functions $\alpha:\{0,1\}^{n}\rightarrow \{0,1\}^{m_\alpha}$, $\beta:\{0,1\}^{n}\rightarrow \{0,1\}^{m_\beta}$ such that $f(x,y)=F(\alpha(x),\beta(y))$. 
\end{definition}
We say that the complexity after pre-processing (with respect to some measure of complexity) of a function $f(x,y)$ is the minimal complexity of any interaction part of $f(x,y)$.
More concretely, for span and branching program size we define the following pre-processed complexity measures. 
\begin{definition} The \textbf{pre-processed branching program complexity} is defined as
    \begin{align}
        BP_{p,(2)}(f)= \min_{F,\alpha,\beta} \{ BP_p(F):f(x,y)=F(\alpha(x),\beta(y))\},
    \end{align}
\end{definition}
\begin{definition} The \textbf{pre-processed span program complexity} is defined as
    \begin{align}
        SP_{p,(2)}(f)= \min_{F,\alpha,\beta} \{ SP_p(F):f(x,y)=F(\alpha(x),\beta(y))\},
    \end{align}
\end{definition}
The pre-processed branching and span program complexities are related polynomially, because the non pre-processed complexities are. 

We define the following pre-processed complexity classes. 
\begin{definition}
    The complexity class $Mod_kL_{(2)}$ is defined as those functions $f:\{0,1\}^{n}\times \{0,1\}^{n}\rightarrow \{0,1\}$ with an interaction part that can be computed with a polynomial size (in n) modulo-p branching program. 
\end{definition}
\begin{definition}
    The complexity class $C_=L_{(2)}$ is defined as those functions $f:\{0,1\}^{n}\times \{0,1\}^{n}\rightarrow \{0,1\}$ with an interaction part that can be computed according to the following procedure.  
    We consider a branching program of polynomial (in n) size. 
    If $acc(x)=rej(x)$, output $1$ and otherwise output $0$.
\end{definition}
We can analogously define the complexity class $P_{(2)}$ as those families of function families which have a poly-time computable interaction part. 

\subsection{Efficiency of protocols for PSM, CDS, and related primitives}

\vspace{0.2cm}
\noindent \textbf{PSM and PSQM protocols}
\vspace{0.2cm}

The largest class of functions for which efficient PSM protocols have been constructed are those with polynomial sized modulo-$p$ branching programs. 
The following theorem was proven in \cite{ishai1997private}.
\begin{theorem}\label{thm:PSMfrombranchingprogram}\textbf{[IK '97]}
    Let $p$ be a prime, and let $\mathbfcal{BP}=(G,\phi,s,t_0,t_1)$ be a Boolean modulo-$p$ branching program of size $a(n)$ computing an interaction part of $f$. Then there exists a PSM protocol for $f$ with randomness complexity and communication complexity both $O( a(n)^2\,\log p )$.
\end{theorem}
Note that the original statement of this theorem considers $f$ rather than its interaction part, but the extension is trivial. 
An immediate consequence of this theorem, along with the implications summarized in figure \ref{fig:web}, is that CDS, PSQM, CDQS, and $f$-routing can all be achieved with the randomness and communication complexity given in the same way, up to constant factor overheads. 

To better understand the implications of this theorem, it is helpful to understand which complexity classes can be efficiently achieved. 
Fixing $p$, those functions with polynomial sized branching programs are exactly the class $Mod_pL$. 
Running the PSM protocol on the interaction part, we can therefore achieve the class $\text{Mod}_pL_{(2)}$ efficiently as a PSM.  
We can also choose $p$ adaptively, and doing so achieve the class $C_=L_{(2)}$. 
This is shown in \cite{ishai1997private}.
It is also interesting to find a complexity class that contains all of those functions where $(\log p) BP_p(f)$ can be made polynomial. 
The smallest class which we can show contains all such functions is $L^{\#L}$, which we state as the following remark. 
\begin{remark}
    Every function family $\{f_n\}$ for which $(\log p)\cdot BP_p(f_n)$ is polynomial in $n$ for some choice of $p$ is contained in the class $L^{\#L}/ poly$.
\end{remark}
\begin{proof}\,
By assumption, there is a polynomial sized branching program, call it $BP$ and denote its size by $s$, whose number of accepting paths counted mod $p$ is non-zero if $f(x)=1$, and $0$ otherwise.  
Further, the choice of $p$ needed must have $\log p$ be polynomial. 
Our algorithm to compute $f$ in $L^{\#L}$ is as follows. 
We take our advice string to be a description of the branching program BP. 
We give BP along with the input $x$ to the $\# L$ oracle, and it will return the number of accepting paths of this program, call it $N$. 
Notice that $N < 2^s$, since there must be no more accepting paths then there are subsets of vertices in BP.
This means the output of the oracle consists of at most a polynomial sized string. 
We then subtract $p$ from $N$ repeatedly until it obtains a number less than $p$. 
Since $p$ also consists of a polynomial number of bits, this can be done in log space. 
\end{proof}

To relate $L^{\#L}$ to more familiar classes, we can note that it is contained inside of DET which is in turn contained inside of NC, where NC is the class of functions computed by poly-logarithmic depth circuits. 

Notice that from theorem \ref{thm:PSMgivesPSQM} the result of theorem \ref{thm:PSMfrombranchingprogram} carries over immediately to the setting of PSQM. 
We move on to understand the implications of theorem \ref{thm:PSMfrombranchingprogram} for the CDS, CDQS, and $f$-routing primitives below. 

\vspace{0.2cm}
\noindent \textbf{CDS protocols}
\vspace{0.2cm}

From theorem \ref{thm:PSMfrombranchingprogram} and because PSM protocols give CDS protocols (see theorem \ref{thm:PSMgivesCDS}), we obtain the following corollary. 
\begin{theorem}\label{thm:CDSefficiencyfromPSM}
    Let $p$ be a prime, and let $\mathbfcal{BP}=(G,\phi,s,t_0,t_1)$ be a Boolean modulo-$p$ branching program of size $a(n)$ computing $f$.
    Then there exists a CDS protocol for $f$ with randomness complexity and communication complexity both $O( a(n)^2\,\log p )$.
\end{theorem}
Note that the implication from PSM to CDS was known already, so that this implication was already clear. Recently, this scaling was improved to linear in the branching program size \cite{ishai2014partial}. 

We can compare this to the most efficient CDS constructions in the literature. 
A CDS protocol based on secret sharing schemes was given in \cite{GERTNER2000592}. 
They prove the following theorem\footnote{The cost here being $c+|s|$ while the cost in the reference \cite{GERTNER2000592} being $c$ is due to our defining the CDS to have the secret held on only one side, rather than on both as is the convention in \cite{GERTNER2000592}.}.
\begin{theorem}\label{thm:CDSfromsecretsharing}\textbf{[GIKM '98]}
Let $h_M:\{0,1\}^{n}\rightarrow \{0,1\}$ be a monotone Boolean function, and let $h:\{0,1\}^n\rightarrow \{0,1\}$ be a projection of $h_M$; that is, $h(y_1,...,y_n)=h_M(g_1,...,g_M)$, where each $g_i$ is a function of a single variable $y_i$. 
Let $S$ be a secret sharing scheme realizing the access structure $h_M$, in which the total share size is $c$, and let $s$ be a secret that can be hidden in $S$. 
Then there exists a protocol $P$ for disclosing $s$ subject to the condition $h$ whose communication and randomness complexity are bounded by $c+|s|$.
\end{theorem}
Using the span program based constructions of secret sharing schemes \cite{karchmer1993span}, this upper bounds the CDS cost of $f$ by the minimal size of a monotone span program computing any projection of $f$, call it $f_M$.
If the span program is over the field $\mathbb{Z}_p$, the cost is $(\log p) \cdot mSP(f_M)$.
In \cite{cree2022code} (see lemma 5) it is shown that the size of a span program computing the projection $f_M$ is the same as the size of a (non-monotone) span program computing $f$, up to a constant additive term. 
This leads to the following corollary. 
\begin{corollary}
    The randomness and communication complexity to perform CDS on the function $f$ is at most $O(\log p\cdot SP_p(f))$, where $SP_p(f)$ is the size of any span program over $\mathbf{Z}_p$ computing $f$. 
\end{corollary}
Notice that this is quite similar to corollary \ref{thm:CDSefficiencyfromPSM}.
Because the span program size and branching program size are related by equation \ref{eq:SPvsBP}, the secret sharing based construction for CDS is always more efficient than the branching program based approach inherited from PSM. 

Another protocol based on dependency programs \cite{pudlak1996algebraic} was given in \cite{applebaum2017private}. 
Because dependency programs are always larger than span programs (see \cite{pudlak1996algebraic}, Lemma 3.6)\footnote{This is true when considering binary inputs, which we do here. The construction in \cite{applebaum2017private} extends to non-binary inputs, and in that setting there may be polynomial overheads.}, the span program based construction remain the most efficient.

\vspace{0.2cm}
\noindent \textbf{CDQS and $f$-routing protocols}
\vspace{0.2cm}

Notice that efficient CDQS protocols are given by both efficient CDS protocols (theorem \ref{thm:CDStoCDQS}) and by PSQM protocols (theorem \ref{thm:PSQMtoCDQS}). 
Further, from theorem \ref{thm:CDQSandfRouting} we have that efficient CDQS leads to efficient $f$-routing.
These implications lead to the following theorem. 
\begin{theorem}\label{thm:CDQSandfRefficiency}
    The randomness and communication complexity to perform CDQS or $f$-routing on the function $f$ is at most $O(\log p\cdot SP_p(f))$. 
\end{theorem}
Since it had not previously been studied in the literature, this gives the largest known class of functions that can be implemented efficiently for CDQS. 

We can compare theorem \ref{thm:CDQSandfRefficiency} to the most efficient protocols known for $f$-routing. In \cite{cree2022code}, the authors proved an upper bound of $O(\log p\cdot SP_p(f))$ on communication and entanglement complexity of $f$-routing, exactly matching the result inherited from classical CDS.  
It is also interesting to note that the protocol given in \cite{cree2022code} that achieves this bound is a close quantum analogue of the CDS protocol devised in the classical setting in \cite{GERTNER2000592}: both protocols are based on storing the secret in a secret sharing scheme and sending or not sending shares based on the value of bits of the input. 

\section{New lower bounds}\label{sec:newlowerbounds}

\subsection{Linear lower bounds on CFE}

We have the following theorem from \cite{kawachi2021communication}. 
\begin{theorem}\textbf{[KN 2021]}\label{thm:PSQMlowerbound}
For a $(1-o(1))$ fraction of functions $f_n : \{0, 1\}^n \times \{0, 1\}^n \rightarrow \{0, 1\}$, the communication complexity of two-party PSQM protocols with shared randomness for $f_n$ is at least $3n-2\log n-O(1)$.
\end{theorem}
In theorem \ref{thm:CFEtoPSQM}, which shows CFE$\rightarrow$PSQM, we could replace shared entanglement in the CFE protocol and obtain a PSQM protocol that only uses shared randomness. 
In fact, the theorem gives that the resulting PSQM uses the same distributed resource state as the CFE. 
From this, theorem \ref{thm:PSQMlowerbound} above gives the following. 
\begin{corollary}\label{corollary:CFElowerbound}
For a $(1-o(1))$ fraction of functions $f_n : \{0, 1\}^n \times \{0, 1\}^n \rightarrow \{0, 1\}$, the communication complexity of coherent function evaluation protocols with shared randomness for $f_n$ is at least $3n-2\log n-O(1)$.
\end{corollary}
Note that we would expect no amount of shared random bits to suffice for a CFE, and instead for entangled states to be required. 
Thus the consequence of this theorem is very weak in the CFE context.  

\subsection{Linear lower bounds on CDQS}

We have the following theorem from \cite{bluhm2021position}.
\begin{theorem}\textbf{[BCS 2022, random function]}.
Let $n\geq 10$. Assume inputs $x,y\in \{0,1\}^n$ are chosen at random. Then there exists a function $f:X\times Y\rightarrow Z$ with $X, Y\in \{0,1\}^n$, $Z\in \{0,1\}$ such that, if the number $q$ of qubits each of the attackers controls satisfies
\begin{align}
    q\leq n/2 - 5
\end{align}
then the attackers are caught with probability at least $2\times 10^{-2}$. 
Moreover, a uniformly random function will have this property, except with exponentially small probability. 
\end{theorem}

Combining this result with theorem \ref{thm:CDQSandfRouting}, we find the following result for CDQS. 
\begin{corollary}\label{corollary:CDQSlowerbound}
There exists a function $f:X\times Y\rightarrow Z$ with $X, Y\in \{0,1\}^n$, $Z\in \{0,1\}$ such that a CDQS protocol which is $\epsilon$-correct and $\delta$-secure for $f$ with $\max\{ \epsilon,\sqrt{\delta} \}< 2\times 10^{-2}$ requires Alice and Bob have a quantum resource system consisting of at least $n/2-5$ qubits.  
Moreover, a uniformly random function will have this property, except with exponentially small probability. 
\end{corollary}

Now applying theorem \ref{thm:PSQMtoCDQS} we obtain the following linear lower bound on the dimension of the resource system in PSQM. 
Note that previously a $2n-O(\log n)$ linear lower bound on communication complexity was known, but no bound on shared entanglement was previously known. 
\begin{corollary}\label{corollary:PSQMlowerbound}
There exists a function $f:X\times Y\rightarrow Z$ with $X, Y\in \{0,1\}^n$, $Z\in \{0,1\}$ such that an $\epsilon$-correct and $\delta$-secure PSQM protocol for $f$ with $\max\{2 \epsilon,\sqrt{2\delta} \} < 2\times 10^{-2}$ requires Alice and Bob have a quantum resource system consisting of at least $n/2-5$ qubits.  
Moreover, a uniformly random function will have this property, except with exponentially small probability. 
\end{corollary}

In the same paper \cite{bluhm2021position}, the author's prove the following bound for the inner product function. 

\begin{theorem}\textbf{[BCS 2022, Inner product]}
Let $n\geq 10$. Assume inputs $x,y\in \{0,1\}^n$ are chosen at random. Then if the number $q$ of qubits each of the attackers controls satisfies
\begin{align}
    q\leq \frac{1}{2}\log n - 5
\end{align}
then the attackers are caught with probability at least $2\times 10^{-2}$ when the function $f$ is chosen to be the inner product function. 
\end{theorem}

This immediately leads to two corollaries analogous to the above, but now with a logarithmic bound and a random function replaced with the inner product. 

\begin{corollary}\label{corollary:CDQSlogbound}
A CDQS protocol for the inner product function on strings of length $n$ which is $\epsilon$-correct and $\delta$-secure with $\max\{ \epsilon,\sqrt{\delta} \}< 2\times 10^{-2}$ requires Alice and Bob have a quantum resource system consisting of at least $\frac{1}{2}\log n - 5$ qubits.  
\end{corollary}

\begin{corollary}\label{corollary:PSQMlogbound}
A PSQM protocol for the inner product function on strings of length $n$ which is $\epsilon$-correct and $\delta$-secure with $\max\{2 \epsilon,\sqrt{2\delta} \} < 2\times 10^{-2}$ requires Alice and Bob have a quantum resource system consisting of at least $\frac{1}{2}\log n-5$ qubits.  
\end{corollary}

\section{New protocols}\label{sec:newprotocols}

\subsection{\texorpdfstring{$f$}{TEXT}-routing for problems outside P/poly}

As discussed in section \ref{sec:complexity}, all general constructions of CDS and PSM only efficiently implement functions inside of the class $(L^{\#L})_{(2)}$. 
As we now discuss, there is a special function which is believed to be outside of $P$ but which has an efficient CDS, CDQS, and $f$-routing protocol. 
This function is known to be at least as hard as the quadratic residuosity problem modulo a composite of unknown factorization.
This efficient protocol is inherited from remark \ref{thm:robustCDSfromSS}, which gives that efficient secret sharing schemes give efficient CDS protocols, along with a non-linear secret sharing scheme constructed in \cite{beimel2005power}.
A less strong, but also interesting construction of a function outside of $L^{\#L}$ with an efficient PSM, CDS, CDQS, and $f$-routing scheme is based on a DRE for the quadratic residuosity problem modulo a prime. 
This function is inside of $P$ but believed to be outside of $NC$. 

We give the two constructions below. 

\subsubsection*{$f$-routing for a problem outside $P$ from non-linear secret sharing}

We define the computational problem that will interest us here. 
\begin{definition}
    The \textbf{quadratic residuosity problem} $QR(u,v)$ is defined as follows. 
    \begin{itemize}
        \item \textbf{Input:} Two integers $u$ and $v$ of $n$ bits.
        \item \textbf{Output:} $1$ if $gcd(u,v)=1$ and there exists an $r$ such that $u=r^2$ mod $v$, and $0$ otherwise.
    \end{itemize}
\end{definition}
The quadratic residuosity function is believed to be outside of $P/poly$.
Its hardness is the basis of a well studied public-key cryptosystem \cite{goldwasser2019probabilistic}, and other cryptographic constructions \cite{cocks2001identity,blum1986simple}. 

For linear secret sharing schemes, it is known that efficient schemes have complexity in the class $Mod_{k}L$ when the scheme is defined over the field $\mathbb{Z}_k$ for $k$ prime. 
Thus the connection from secret sharing to CDS to CDQS and $f$-routing reproduces the known class of functions that can be efficiently implemented in the $f$-routing setting. 

Beyond linear schemes, \cite{beimel2005power} constructed secret sharing schemes with indicator functions that have complexity outside of $P$. 
Their scheme realizes the following access structure. 
\begin{definition}
    \textbf{NQR}$_n$ is an access structure on $n=4m$ parties for $m$ an integer. 
    We label the $4m$ shares by $W^b_i$ and $U^b_j$ with $b\in \{0,1\}$ and $j\in\{1,...,m\}$. 
    Given two bit strings\footnote{To do modular arithmetic with $w,u$ numbers in $\{ 0, \dots, 2^m-1 \}$ are associated to the strings $w,u$.} $w,u$ each of length $m$, we associate a subset $B_{w,u}$ of size $2m$ according to
    \begin{align}
        B_{w,u} = \{W^{w_i}_{i}: 1 \leq i\leq m \}\cup \{U^{u_i}_{i}: 1 \leq i\leq m \}.
    \end{align}
    The access structure \textbf{NQR}$_n$ is then defined by its minimal authorized sets, which are
    \begin{itemize}
        \item $\{W_i^0,W_i^1\}$ for any $1\leq i\leq m$
        \item $\{U_i^0,U_i^1\}$ for any $1\leq i\leq m$
        \item $B_{w,u}$ for $w,u$ such that $u\neq 0,1$ and $QR(w,u)=0$, so that $w$ is not a quadratic residue modulo $u$. 
        \item $B_{w,u=0}$ for $w\neq 1$.
    \end{itemize}
\end{definition}

Evaluating the indicator function for this access structure is at least as hard as solving the quadratic residuosity problem. 
To see this, notice that we can reduce computing $QR(u,w)$ to evaluating $f_I$ as follows. 
From the string $w$ of length $m$, define the two strings $\tilde{w}$, $\tilde{w}'$ according to 
\begin{align}
    \tilde{w}_i = \begin{cases}
		1 & \text{if} \,\,\,\, w_i=1 \\
            0 & \text{otherwise}
		 \end{cases}
\end{align}
\begin{align}
    \tilde{w}_i' = \begin{cases}
		1 & \text{if} \,\,\,\, w_i=0 \\
            0 & \text{otherwise}
		 \end{cases}
\end{align}
We similarly define $\tilde{u}$ and $\tilde{u}'$, and then notice that
\begin{align}
    QR(w,u) = \neg f_I(\tilde{w},\tilde{w}',\tilde{u},\tilde{u}')
\end{align}
Since computing $\tilde{w},\tilde{w}',\tilde{u},\tilde{u}'$ from $(w,u)$ can be done efficiently, computing $f_I$ is not harder than computing $QR(w,u)$. 

Despite the indicator function being of high complexity, there exists an efficient secret sharing scheme for the access structure \textbf{NQR}$_n$.
This is given in the following theorem. 
\begin{theorem}\label{thm:biemel}
    \textbf{[BI 2005]} There exists an $\epsilon$ secure and $\delta$ private secret sharing scheme for the access structure \textbf{NQR}$_{n}$ storing a single bit secret with security parameter $k$, and
    \begin{itemize}
        \item share size $O(k^2+km)$,
        \item correctness $\epsilon = 2^{-k}$,
        \item security\footnote{Note that our security definition in terms of a simulator is different from the definition in \cite{beimel2005power}, but it is straightforward to show their security definition with value $\delta$ implies ours with the same $\delta$.} $\delta= k/2^k$.
    \end{itemize}
\end{theorem}
We refer the reader to \cite{beimel2005power} for the construction of this scheme. 

In the context of these distributed cryptographic tasks, we are interested in functions which remain of high complexity even when allowing for pre-processing. 
Thus we would like to construct functions outside of $P_{(2)}$, perhaps starting with NQR. 
For a function to be a likely candidate to be outside $P_{(2)}$, we need to ensure pre-processing is as unhelpful as possible. 
We suggest the following function
\begin{align}
    NQR_{4m,(2)}(x,y) = NQR_{4m}(x\oplus y)
\end{align}
Then, since Alice see's only $x$ and Bob see's only $y$, pre-processing seems no better than advice, so we expect that NQR$_{4m,(2)}$ is outside $P_{(2)}$ if we have that $NQR_{4m}$ is outside $P/poly$, as we commented above is believed. 
We state this as the following assumption. 
\begin{conjecture}\label{conj:hardness}
    The function NQR$_{4m,(2)}(x,y)$ is outside of $P_{(2)}$.
\end{conjecture}
Next, we claim that there is an efficient CDS scheme for NQR$_{4m,(2)}(x,y)$.
To see this, we have Alice, following remark \ref{thm:robustCDSfromSS}, prepare the scheme in theorem \ref{thm:biemel} with access structure NQR$_{4m}(z)$.
Then she takes share $S_i$ to be the secret which will be conditionally disclosed in a scheme on the XOR function with inputs $x_i$ and $y_i$. 
Correctly implementing each of these CDS schemes for the shares $S_i$ is easily seen to now correctly implement the larger scheme with access structure NQR$_{(2),4m}$.
This CDS can be performed using $O(|S_i|)$ randomness, so the total needed randomness is still given by the size of the secret sharing scheme. 

From this construction for CDS and theorem \ref{thm:CDStoCDQS} we obtain the following. 
\begin{corollary}\label{corolary:CDSandCDQSoutsideP}
    Assuming conjecture \ref{conj:hardness}, there exists a function outside of $P_{(2)}$ with $n$ input bits and hiding one (qu)bit for which CDS and CDQS can be performed $\epsilon=2^{-k}$ correctly and $\delta=k 2^{-k}$ securely with $O(k^2+kn)$ shared bits of randomness.
\end{corollary}
From theorem \ref{thm:CDQSandfRouting}, we then obtain the following consequence for $f$-routing.
\begin{corollary}\label{corollary:fRouteoutsideP}
    Assuming conjecture \ref{conj:hardness}, there exists a function outside of $P_{(2)}$ with $n$ input bits and hiding one (qu)bit for which f-routing can be performed $\epsilon=O(k2^{-k})$ correctly with $O(k^2+kn)$ shared entangled pairs.
\end{corollary}

\subsubsection*{$f$-routing for a problem outside NC from DRE}

Next, we construct a CDS scheme for a lower complexity function, albeit one that is still outside of $NC$, via a second route that begins with a decomposable randomized encoding.\footnote{Another route for a construction of an $f$-routing scheme for a problem outside NC but inside P, and which is exact, is to begin with (exact) non-linear secret sharing scheme given in \cite{beimel2005power}. We've chosen to use a route beginning with DRE to illustrate that interesting connection.}
The computational problem that will interest us is again quadratic residuosity, but this time where the modulus is taken over a prime.  
\begin{definition}
    The \textbf{quadratic residuosity problem over $\mathbb{Z}_p$} is defined as follows. 
    \begin{itemize}
        \item \textbf{Input:} An integer $a$ of $n$ bits and prime $p$, also of $n$ bits.
        \item \textbf{Output:} $1$ if $a=b^2$ mod $p$ for some $b$, and $0$ otherwise.
    \end{itemize}
\end{definition}
While this problem is not known to be inside of NC, but is easily placed inside of $P$ by recalling the Euler criterion, which states that
\begin{align}
    a^{\frac{p-1}{2}} = 1 \,\, \text{mod}\,\, p
\end{align}
if and only if $a$ is a square. 
Given this, modular exponentiation can be used to determine if $a$ is a square in polynomial time. 
Note that if we pose the same problem but with the prime $p$ replaced by a composite number the resulting problem is thought to be outside of $P$ \cite{Kaliski2011}. 
We focus on the prime case here. 
See \cite{beimel2005power} for a related discussion of the complexity of the quadratic residuosity functions considered over a field $\mathbb{Z}_p$ for $p$ prime. 

The quadratic residuosity problem over primes admits a simple randomized encoding scheme. 
In particular take
\begin{align}
    a \rightarrow r^2 a
\end{align}
for $r$ a randomly chosen integer in $\mathbb{Z}_p$. 
To understand why this is a randomized encoding, notice that $QR(a)=QR(r^2a)$, so we can compute the result of the function defined by the residuosity problem from the encoded output correctly, by (in this particular case) simply computing the original function, since $r^2a$ is a quadratic residue if $a$ is. 
Next, to show security one needs to show that if $a$ is a quadratic residue then $r^2a$ is randomly distributed over all those integers $\tilde{a}$ in $\mathbb{Z}_p$ which also are, and if $a$ is not a quadratic residue then $r^2$ is uniformly distributed over all those $\tilde{a}$ which are also not. 
This amounts to showing that if $a$ and $\tilde{a}$ both are (or both are not) quadratic residues then there is a unique $r$ such that $r^2 a=\tilde{a}$. 
This follows because the product of two residues is a residue, and the product of two non-residues is a residue.

We can further extend this to a decomposable randomized encoding as follows \cite{ball2020complexity}. 
Use the encoding
\begin{align}
    a_i \rightarrow a_i r^2 2^{i-1} + s_i =: y_i
\end{align}
for $s_i$, $r$ drawn independently and at random from $\mathbb{Z}_p$ for all but the last $s_i$, which we set so that $\sum_i s_i=0$. 
Then to decode use
\begin{align}
    QR\left( \sum_i y_i \right) = QR(r^2a) = QR(a).
\end{align}
To see security, we assume that $a$, $\tilde{a}$ are two integers with the same quadratic residue, and then show there is a choice of $r$, $s_i$ which make the bits of $a$ look like the bits of $\tilde{a}$. 
This means we need to solve
\begin{align}
    a_i 2^{i-1} = \tilde{a}_i2^{i-1}r^2 + s_i
\end{align}
subject also to $\sum_i s_i=0$. 
It is easy to see we can do this taking as an assumption the same thing we used in the earlier case, that if $a$, $\tilde{a}$ have the same quadratic residue then there is a $r$ such that $a=r^2\tilde{a}$. 

Given the existence of a decomposable randomized encoding scheme for the quadratic residue problem, we immediately obtain a PSM for this problem as noted above: Alice and Bob simply send the randomized encodings of their input bits to the referee, who runs the decoding procedure. 
This was observed already in \cite{ishai1997private}. 
This in turn implies an efficient CDS, CDQS, and $f$-routing scheme for $f(x)=QR(x)$. 
We collect these observations as the following remark. 
\begin{remark}\label{remark:outsideNC}
    Consider an $n$ bit string $z$ and split its bits into arbitrary subsets $S$ and $S^c$.
    Let the bits from $S$ define a string $z_S$ and a bit from $S^c$ define a string $z_{S^c}$. 
    Then the function $f(z_S,z_{S^c})=QR(z)$ has perfectly correct PSM and CDS schemes that uses poly$(n)$ bits of randomness. 
\end{remark}
We can also use theorems \ref{thm:CDStoCDQS} and \ref{thm:CDQSandfRouting} to upgrade these to quantum schemes, giving the following corollary. 
\begin{corollary}
    Consider an $n$ bit string $z$ and split its bits into arbitrary subsets $S$ and $S^c$.
    Let the bits from $S$ define a string $z_S$ and a bit from $S^c$ define a string $z_{S^c}$. 
    Then the function $f(z_S,z_{S^c})=QR(z)$ have perfectly correct PSQM and CDQS schemes that uses poly$(n)$ EPR pairs as a resource state. 
\end{corollary}
Ideally, one would show that, assuming $QR(z)$ is outside of NC implies $f(z_{S},z_{S^c})$ is outside of NC$_{(2)}$ but we are unable to do so. 
Nonetheless, this constructs a second problem not known to be in NC$_{(2)}$ with an efficient $f$-routing scheme, although this one is inside of P. 
Another comment is that this problem has an exact scheme, while the construction in the previous section that is outside of P is approximate. 

\subsection{Efficient PSQM and CDQS for low T-depth circuits}

In \cite{speelman2015instantaneous}, a protocol is given that performs a unitary $\mathbf{U}_{AB}$ non-locally with entanglement cost that depends on the circuit decomposition of $\mathbf{U}_{AB}$. 
In particular we write $\mathbf{U}_{AB}$ in terms of a Clifford + T gate set, and obtain the following two upper bounds on entanglement cost. 

\begin{theorem}\label{thm:Tcountprotocol}
    Any $n$ qubit Clifford + $T$ quantum circuit $C$ which has at most $k$ $T$-gates can be implemented non-locally using $O(n2^k)$ EPR pairs. 
    Further, if $C$ has $T$-depth $d$ then there is a protocol to implement $C$ non-locally using $O((68 n)^d)$ EPR pairs. 
\end{theorem}

From theorems \ref{thm:CDQSandfRefficiency} and \ref{thm:CFEtoPSQM}, these results lead to upper bounds on entanglement cost in implementing CDQS, $f$-routing, and PSQM. 
These upper bounds depend on the number of $T$ gates needed to compute $f(x,y)$ with a quantum circuit. 
We discuss the CDQS setting first. 

\begin{figure*}
    \centering
    \begin{quantikz}
    \lstick[2]{x} &\gate[4]{\mathbf{U}}& \qw & \qw & \\
    & & \qw & \qw & \\
    \lstick[2]{y}& & \qw & \qw & \\
    & & \ctrl{2} & \qw & \\
    A_0 & \qw & \swap{1} & \qw & \\
    A_1 & \qw & \swap{-1} & \qw & 
    \end{quantikz}
    \caption{The circuit implementing the unitary $\mathbf{U}'$. The unitary $\mathbf{U}$ computes $f(x,y)$ on its last wire with high fidelity. System $A_0$ is initially maximally entangled with reference $R$. At the end of the circuit, $R$ with be highly entangled with system $A_{f(x,y)}$.}
    \label{fig:UwithSWAP}
\end{figure*}
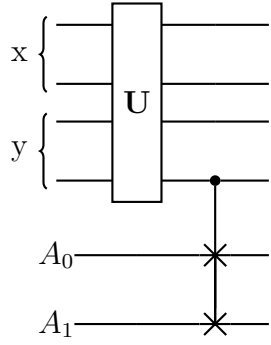

\begin{corollary}\label{corollary:Tdepthandf-route}
    Suppose that a function $f(x,y)$ can be evaluated with probability $1-\epsilon$ by a Clifford + T circuit with T-count k and T-depth d. 
    Then there is a $2\epsilon$-correct $f$-routing protocol for the function $f(x,y)$ that uses at most $O(n2^k)$ EPR pairs, or at most $O((68n)^{d+5})$ EPR pairs, whichever is smaller.
\end{corollary}
\begin{proof}\,
Let $\mathbf{U}$ be the unitary that computes $f$. 
Recall that this means a measurement in the computational basis on the first qubit of the output of $\mathbf{U}$ returns $f(x,y)$ with probability $1-\epsilon$. 
Writing the state
\begin{align}
    \mathbf{U}\ket{x,y} &= \sum_{i_2,...,i_n} \alpha_{0,i_2,...,i_n} \ket{0}\ket{i_2,...,i_n} + \sum_{i_2,...,i_n} \alpha_{1,i_2,...,i_n} \ket{1}\ket{i_2,...,i_n} \nonumber\\
    &= \alpha_0 \ket{\psi_0} + \alpha_1 \ket{\psi_1}
\end{align}
we have that $|\alpha_{f(x,y)}|^2 \geq 1 - \epsilon$. 

Now consider modifying the circuit that implements $\mathbf{U}$ by adding two ancilla qubits $A_0A_1$ and a controlled SWAP gate, where we control on the first output qubit of $\mathbf{U}$. 
We show this as a quantum circuit in figure \ref{fig:UwithSWAP}. 
The controlled SWAP gate can be implemented with 7 $T$-gates arranged in $5$ layers (see e.g. \cite{kim2018efficient}).
Thus our new circuit has $T$-depth at most $d+5$ and $T$-count at most $k+7$. 
We call the unitary $\mathbf{U}$ composed with the controlled swap gate $\mathbf{U}'$. 

To implement the $f$-routing protocol, we implement $\mathbf{U}'$ non-locally with $A_0X$ held on the left and $A_1Y$ held on the right. 
Initially $A_0$ is in the maximally entangled state with the reference system $R$. 
Because $\mathbf{U}'$ can be implemented with $k+7$ $T$-gates and $T$-depth of $d+5$, theorem \ref{thm:Tcountprotocol} gives that this takes no more than $O(n2^k)$ EPR pairs, or at most $O((68n)^{d+5})$ EPR pairs, whichever is smaller.
Then we claim that at the end of the protocol that the $A_{f(x,y)}$ system is nearly maximally entangled with $R$. 

To see this, notice that the state of the $RA_0A_1XY$ after the unitary plus controlled swap have been applied is
\begin{align}
    \alpha_0\ket{\Psi^+}_{RA_{0}}\ket{0}_{A_1}\ket{\psi_0}_{XY} + \alpha_1\ket{\Psi^+}_{RA_{1}}\ket{0}_{A_0}\ket{\psi_1}_{XY},
\end{align}
where $\psi_0$ and $\psi_1$ are orthogonal states as a consequence of unitarity of $\mathbf{U}$. 
We take the decoding channel to be the trace over the $A_{1-f(x,y)}XY$ system, followed by a relabelling of $A_{f(x,y)}$ as $Q$. 
This produces the state
\begin{align}
    \rho_{RQ} = |\alpha_{f(x,y)}|^2 \Psi^+_{RQ} + |\alpha_{1-f(x,y)}|^2 \frac{\mathcal{I}}{d_R}\otimes \ketbra{0}{0}_{Q}
\end{align}
Then we can calculate the fidelity
\begin{align}
    F(\Psi^+,\rho_{RQ}) \geq |\alpha_{f(x,y)}|^2 \geq 1-2\epsilon
\end{align}
so that the $f$-routing protocol is $2\epsilon$ correct, as needed. 
\end{proof}

From theorem \ref{thm:CDQSandfRefficiency}, this also leads to a similar upper bound for CDQS. 
\begin{corollary}\label{corollary:TdepthandCDQS}
    Suppose that a function $f(x,y)$ can be evaluated with probability $1-\epsilon$ by a Clifford + T circuit with T-count k and T-depth d. 
    Then there is a $2\epsilon$-correct and $\sqrt{\epsilon \log d_Q}$ secure CDQS protocol for the function $f(x,y)$ that uses at most $O(n2^k)$ EPR pairs, or at most $O((68n)^d n^{5})$ EPR pairs, whichever is smaller.
\end{corollary}
\begin{proof}\,
    Immediate from theorem \ref{thm:CDQSandfRouting}. 
\end{proof}

Next, we apply theorem \ref{thm:Tcountprotocol} to give a class of functions for which PSQM can be efficiently performed. 
\begin{corollary}\label{corollary:PSQMandTdepth}
    Suppose that the isometry 
    \begin{align}
        \mathbf{V}_f = \sum_{xy} \ket{xy}_{Z'} \ket{f_{xy}}_{Z} \bra{x}_X\bra{y}_Y
    \end{align}
    can be implemented with closeness $\epsilon$ (according to the diamond norm distance) with a Clifford + T circuit with T-count k and T-depth d. 
    Then there exists a PSQM protocol for $f(x,y)$ which is $\epsilon$-correct and $\sqrt{\epsilon}$-secure that uses at most $O(n2^k)$ EPR pairs, or at most $O((68n)^d n^{5})$ EPR pairs, whichever is smaller.
\end{corollary}
\begin{proof}\,
    Follows immediately from theorems \ref{thm:CFEtoPSQM} and \ref{thm:Tcountprotocol} taken together. 
\end{proof}

\subsection{Sub-exponential protocols for \texorpdfstring{$f$}{TEXT}-routing on arbitrary functions}

In a surprising breakthrough, \cite{liu2017conditional} showed that CDS can be performed for any function using sub-exponential communication and randomness. 
We summarize their result as the following theorem. 

\begin{theorem}\label{thm:LVW} \textbf{[LVW 2017]}
    Every function $f:\{0,1\}^{n}\times \{0,1\}^{n}\rightarrow \{0,1\}$ has a CDS protocol for single bit secrets using $2^{O(\sqrt{n\log n})}$ bits of randomness and $2^{O(\sqrt{n\log n})}$ bits of communication. 
\end{theorem}
Combining this with theorem \ref{thm:CDStoCDQS} we obtain the following corollary. 
\begin{corollary} \label{thm:subexpCDQS}
    There exist CDQS protocols with perfect correctness and secrecy for every function $f:\{0,1\}^{n}\times \{0,1\}^{n}\rightarrow \{0,1\}$ using $2^{O(\sqrt{n\log n})}$ bits of randomness and $2^{O(\sqrt{n\log n})}$ bits of communication, along with a single qubit of communication. 
\end{corollary}
\begin{proof}\,
    Recall that CDS protocols for secrets $s_1$, $s_2$ can be run in parallel if using fresh randomness for each instance (see the paragraph after remark \ref{remark:onesidedCDS}). 
    Thus we can create a CDS hiding two bits of secret while still using $2^{O(\sqrt{n\log n})}$ randomness and communication, and then apply theorem \ref{thm:CDStoCDQS} to see that we can perform CDQS on a single qubit. 
\end{proof}

From here, theorem \ref{thm:CDQSandfRouting} leads to the following. 
\begin{corollary}\label{corollary:subexpfroute}
    There exists a perfectly correct $f$-routing protocol for every function $f:\{0,1\}^{n}\times \{0,1\}^{n}\rightarrow \{0,1\}$ using $2^{O(\sqrt{n\log n})}$ qubits of resource system and $2^{O(\sqrt{n\log n})}$ qubits of message. 
\end{corollary}
\begin{proof}\,
    Immediate from corollary \ref{thm:subexpCDQS} and theorem \ref{thm:CDQSandfRouting}.
\end{proof}

Before moving on, we give some brief context for the construction in \cite{liu2017conditional} that leads to sub-exponential CDS protocols. 
The reader interested in the construction may refer to the original reference \cite{liu2017conditional} or to the lectures \cite{BIUschool}. 

The construction begins with a reduction from a CDS protocol for a general function $f(x,y)$ to a particular function we denote as $INDEX(x,D_y)$, which takes as input Alice's input $x$ and the string $D_y=f(00...00,y)f(00...01,y)...f(11...11,y)$ of length $2^n$. 
Notice that
\begin{align}
    f(x,y) = INDEX(x,D_y) = D_y[x]
\end{align}
This means in particular that a good CDS protocol for the index function will lead to a good CDS protocol for all functions. 

The construction of a CDS for INDEX begins with a connection to the cryptographic task of \emph{private information retrieval} (PIR). 
In a PIR task, a client interacts with several non-communicating servers to retrieve an item with label $x$ from a database $D$, call the item $D[x]$. 
Security of the PIR requires that the databases not be able to determine the label $x$. 
This primitive has long been noted to be related to CDS, and in fact CDS was first defined in the context of studying PIR schemes \cite{GERTNER2000592}. While it is not known if all PIR schemes induce CDS schemes, techniques used in PIR constructions have led to CDS schemes. 
Theorem \ref{thm:LVW} was proven by applying tools from a sub-exponential PIR scheme presented in \cite{dvir20162} to construct a CDS. 

The PIR scheme developed in \cite{dvir20162} relies on the existence of large \emph{matching vector families}. 
A set of pairs of vectors $\{(\mathbf{u}_i,\mathbf{v}_i)\}_{i=1}^N$ is said to be a $S$-matching vector family if
\begin{align}
    \langle \mathbf{u}_i, \mathbf{v}_i \rangle &= 0 \\
    \langle \mathbf{u}_i, \mathbf{v}_j \rangle &\in S,\ \text{when } i \neq j.
\end{align}
Matching vector families find other applications as well, for instance in the construction of locally decodable codes. 
An outstanding question is how large $N$ can be taken for vectors chosen in a given vector space. 
In \cite{grolmusz2000superpolynomial}, the authors constructed large matching vector families over $\mathbb{Z}_6^\ell$, which lead to efficient PIR schemes.  
Using similar techniques, the same matching vector families lead to the efficient CDS scheme of \cite{liu2017conditional}.

\section{Discussion}\label{sec:discussion}

\vspace{0.2cm}
\noindent \textbf{Collapse of CDQS and PSQM complexity with PR boxes}
\vspace{0.2cm}

A Popescu-Rohrlich box is a hypothetical device, shared by distant parties Alice and Bob, which allows them to satsify the CHSH game with probability one. 
More concretely, given input $x$ on Alice's side and input $y$ on Bob's side, the device returns $a$ to Alice and $b$ to Bob such that $a\oplus b = x\wedge y$. 
Broadbent \cite{broadbent2016popescu} showed that if Alice and Bob share PR boxes, they can implement any unitary as a non-local computation using only linear entanglement and a linear number of uses of a PR box. 
This can be seen as a quantum analogue of a similar collapse that occurs in the setting of classical communication complexity \cite{van2013implausible}. 
Because efficient non-local computation  protocols lead, via theorems \ref{thm:CDQSandfRouting} and \ref{thm:CFEtoPSQM}, to efficient CDQS and PSQM protocols, Broadbent's result similarly leads to a collapse to linear cost for PSQM and CDQS. 

In fact, an even stronger collapse follows for CDQS, PSQM and $f$-routing by applying the result of \cite{van2013implausible} showing the collapse of classical communication complexity in the presence of PR boxes. 
In particular, PR boxes can be used to reduce computing $f(x,y)$ with $x$ held by Alice and $y$ held by Bob to computing $\alpha+\beta$, with $\alpha$ computed from $x$ plus the output of PR box uses, and $\beta$ computed from $y$ along with PR box uses.\footnote{See \cite{kaplan2011non} for results on the number of PR box uses necessary. Note that in our setting we can use the PR boxes sequentially if desired.} 
In the CDS or PSM settings then, we need only execute CDS or PSM on the function $g(\alpha,\beta)=\alpha+\beta$ with the inputs being single bits. 
This can be done with $O(1)$ randomness. 
Via theorems \ref{thm:CDStoCDQS} and \ref{thm:PSMgivesPSQM} then, CDQS can be done with $O(1)$ EPR pairs and PSQM with $O(1)$ shared random bits. 
We can further note that from theorem \ref{thm:CDQSandfRouting} this means $f$-routing can be performed for arbitrary functions using only $O(1)$ EPR pairs when given access to PR boxes. 

\vspace{0.2cm}
\noindent \textbf{Connections to quantum gravity and holography}
\vspace{0.2cm}

In the study of quantum gravity the holographic principle \cite{hooft1993dimensional,susskind1995world} asserts that gravity in $d$ dimensions should have an alternative quantum mechanical description in just $d-1$ dimensions. 
This principle is realized manifestly in the context of the AdS/CFT correspondence \cite{maldacena1999large,witten1998anti}. 
In \cite{may2019quantum}, holography and the AdS/CFT correspondence was related to non-local quantum computation. 
In particular, they argued local interactions in the higher dimensional gravity picture are reproduced as non-local quantum computations in the lower dimensional quantum mechanical picture.
As a consequence, computations in the presence of gravity may be constrained by limits on entanglement in the dual quantum mechanical picture \cite{may2022complexity}, or interactions in the gravity picture may imply more computations can be performed non-locally than we have so far found protocols for. 

In this work, we see that as a consequence of their connections to NLQC, CDQS and PSQM are also related to holography. 
One can also realize CDQS and PSQM protocols directly in holography, using connections similar to the one in~\cite{may2019quantum} or the more recent~\cite{may2023non}. 
This implies that, as with NLQC, constraints on CDQS and PSQM correspond to constraints on bulk interactions. 
Conversely, the holographic picture has been argued \cite{may2020holographic,may2022complexity} to suggest that a larger class of unitaries than is currently known should have efficient non-local implementations.
Importantly, the connection between CDQS and PSQM is so far limited to the 2 input player case, which is also the case that ties to NLQC. 
It may be possible to explore a connection between CDQS and PSQM to holography that is realized more directly, not via NLQC, which could extend the connection to settings with many input players. 

Recalling \cite{may2022complexity}, it was argued that the holographic connection suggests that at least unitaries in BQP should be implementable non-locally. 
From this perspective it is interesting that, from the connection to secret sharing, we now have at least one function outside of P but inside of BQP with an efficient non-local implementation. 

\vspace{0.2cm}
\noindent \textbf{Quantum analogues of recent classical results}
\vspace{0.2cm}

Non-local quantum computation was previously thought to have no (non-trivial) classical analogue: taking the inputs and outputs of a computation to be classical, one can immediately perform the computation in the non-local form of figure \ref{fig:non-localcomputation} without use of shared randomness.\footnote{This amounts to a special case of the impossibility result \cite{chandran2009position}. To see why it is true, consider copying the inputs $x$ and $y$ where they are received and forwarding a copy across the communication channel.}
The connections pointed out in this article give non-trivial classical analogues of non-local computation: CDQS is equivalent to a special case of NLQC, and has a non-trivial classical version (CDS), and similarly to PSM. 

Traditionally, classical analogues are a source of techniques and conjectures in the quantum setting. 
Taking this perspective on CDS and CDQS, two recent results in the CDS literature are natural candidates to revisit in the quantum setting. 

First, in \cite{applebaum2021placing}, the authors relate CDS to various communication complexity scenarios. 
In particular they consider the communication complexity class $AM^{cc}$, defined as follows. 
Alice and Bob hold inputs $x$ and $y$ and share randomness $r$, while a referee holds $(x,y)$. 
The referee will send Alice and Bob a proof $p=p(x,y,r)$ that both Alice and Bob should accept when $f(x,y)=1$, and both should reject if $f(x,y)=0$. 
$AM^{cc}(f)$ is the minimal length of the needed proof, and $AM^{cc}$ is the class of functions for which the proof can be taken to be of polylogarithmic length. 
Relating this to CDS, they show that for some constant $c>0$,
\begin{align}\label{eq:CDSandAM}
    CDS(f) \geq (\max\{AM^{cc}(f),coAM^{cc}(f) \})^c - \text{polylog}(n)
\end{align}
where CDS($f$) is the communication complexity of a CDS protocol for $f$ (allowing for imperfect correctness and imperfect security), and a similar bound differing only by constant factors exists for randomness complexity. 
Unfortunately, there are no explicit functions known to be outside $AM^{cc}\cap coAM^{cc}$, but nonetheless equation \ref{eq:CDSandAM} is an intriguing result. 
A natural question is if a similar inequality holds when considering CDQS and quantum communication complexity classes. 

Second, the related work \cite{applebaum2017private} studied the relationship between zero-knowledge proofs and both CDS and PSM. 
The starting point is a zero-knowledge variant of the class $AM^{cc}$ discussed above, where an additional requirement that the proof $p$ not reveal anything about $(x,y)$ is imposed. 
This is refereed to as the class $ZAM^{cc}$. 
The authors of \cite{applebaum2017private} found that a PSM protocol with perfect correctness and privacy leads to a similarly efficient ZAM protocol, and that a ZAM protocol (which may be approximate) leads to a similarly efficient CDS protocol. 
Again, it is natural to ask for a quantum analogue of these results. 

\vspace{0.2cm}
\noindent \textbf{Classical analogues of further non-local computations}
\vspace{0.2cm}

In this paper we relate two special cases of non-local quantum computation --- $f$-routing and coherent function evaluation --- to other cryptographic tasks, CDQS and PSQM. 
One aspect of these relationships we have emphasized is that while non-local computation naively becomes trivial when considered classically\footnote{In particular we have in mind that a non-local computation with only classical inputs can always be implemented without pre-distributed resources \cite{chandran2009position}.}, PSQM and CDQS have natural classical variants. 
This raises the question as to whether NLQC generally has a good classical analogue, perhaps one exploiting the same communication pattern as CDS and PSM, and employing an appropriate secrecy condition. 
Less ambitiously, we can also ask about classical analogues of other commonly studied non-local quantum computation schemes. 
One commonly studied non-local computation which we have not considered here is the BB84 task \cite{buhrman2014position,tomamichel2013monogamy}, and its extension to $f$-BB84 \cite{bluhm2022single,escolàfarràs2022singlequbit}. 
It would be interesting to understand if $f$-BB84 is related to a classical primitive. 

\vspace{0.2cm}
\noindent \textbf{Acknowledgements:} We thank Adrian Kent and David P\'erez-Garcia for helpful discussions.
Thomas Koutsikos pointed out some errors in our proofs which have now been corrected.
RA and HB were supported by the Dutch Research Council (NWO/OCW), as part of the Quantum Software Consortium programme (project number 024.003.037).
PVL and HB were supported by the Dutch Research Council (NWO/OCW), as part of the NWO Gravitation Programme Networks (project number 024.002.003).
FS was supported by the Dutch Ministry of Economic Affairs and Climate Policy (EZK), as part of the Quantum Delta NL programme.
AM is supported by the Simons Foundation It from Qubit collaboration, a PDF fellowship provided by Canada's National Science and Engineering Research council, by Q-FARM, and the Perimeter Institute of Theoretical Physics.
Research at Perimeter Institute is supported in part by the Government of Canada through the Department of Innovation, Science and Economic Development Canada and by the Province of Ontario through the Ministry of Colleges and Universities.

\appendix 

\bibliographystyle{unsrtnat}
\bibliography{biblio}

\end{document}